\documentclass[12pt]{article}

\usepackage{enumerate}
\usepackage{url}

\newcommand{\blind}{0}

\usepackage{graphicx}
\usepackage{tikz}
\usetikzlibrary{arrows,decorations.pathmorphing,backgrounds,positioning,fit,petri}
\usepackage{subcaption}
\usepackage{array}

\usepackage{amsmath}
\usepackage{amsthm}
\theoremstyle{plain}
\newtheorem{definition}{Definition}
\newtheorem{remark}{Remark}
\newtheorem{theorem}{Theorem}
\newtheorem{proposition}{Proposition}
\newtheorem{assumption}{Assumption}
\newtheorem{lemma}{Lemma}
\usepackage{algorithm,algorithmic}

\usepackage{natbib}
\usepackage{hyperref}

\usepackage{naesseth}
\usepackage{abbreviations}

\begin{document}
\def\spacingset#1{\renewcommand{\baselinestretch}%
{#1}\small\normalsize} \spacingset{1}

\if0\blind
{
  \title{\bf High-dimensional Filtering using Nested Sequential Monte Carlo}
  \author{Christian A. Naesseth\footnote{Corresponding author. Address: Department of Electrical Engineering, Link\"oping University, SE-581 83 Link\"oping, Sweden. Email: \href{mailto:christian.a.naesseth@liu.se}{christian.a.naesseth@liu.se}}\hspace{.2cm},
   Fredrik Lindsten\footnote{Department of Information Technology, Uppsala University}\hspace{.2cm} 
   and 
   Thomas B. Sch\"on\footnotemark[\value{footnote}]\hspace{.2cm}
  }
  \maketitle
} \fi

\if1\blind
{
  \bigskip
  \bigskip
  \bigskip
  \begin{center}
    {\LARGE\bf High-dimensional Filtering using Nested Sequential Monte Carlo}
\end{center}
  \medskip
} \fi

\bigskip
\begin{abstract}
Sequential Monte Carlo (\smc) methods comprise one of the most successful approaches to approximate Bayesian filtering. However, \smc without good proposal distributions struggle in high dimensions. We propose nested sequential Monte Carlo (\nsmc), a methodology that generalises
the \smc framework by requiring only approximate, properly weighted, samples from the \smc proposal distribution, while still resulting in a correct \smc algorithm. This way we can \emph{exactly approximate} the locally optimal proposal, and extend the class of models for which we can perform efficient inference using \smc. We show improved accuracy over other state-of-the-art methods on several spatio-temporal state space models.
\end{abstract}

\noindent%
{\it Keywords:}  particle filtering, spatio-temporal models, state space models, approximate Bayesian inference, backward simulation
\vfill

\newpage
\spacingset{1.45}

% ======================================================================
%                           Introduction
% ======================================================================
\section{Introduction}\label{sec:introduction}
Inference in complex and high-dimensional statistical models is a very challenging problem that is
ubiquitous in applications such as climate informatics
\citep{monteleoniEtAl2013}, bioinformatics \citep{cohen2004bioinformatics} and machine learning
\citep{wainwright2008graphical}, to mention a few.

We are interested in \emph{sequential} Bayesian inference in settings where we have a sequence of posterior distributions that we need to compute. To be specific, we are focusing on settings where the model (or state variable) is high-dimensional, but where there are \emph{local dependencies}. One example of the type of models we consider are the so-called spatio-temporal models \citep{Wikle:2015,CressieW:2011,RueH:2005}. 

Sequential Monte Carlo (\smc) methods comprise one of the most successful methodologies for sequential Bayesian inference.
However, \smc struggles in high dimensions and these methods are rarely used for dimensions, say, higher than ten \citep{rebeschiniH2015can}.
The purpose of the \nsmc methodology is to push this limit well beyond the single digits.

The basic strategy is to mimic the behavior of a so-called \emph{fully adapted} (or locally optimal) \smc algorithm.
Full adaptation can drastically improve the efficiency of \smc in high dimensions \citep{snyder2015performance}. Unfortunately, it can rarely
be implemented in practice since the fully adapted proposal distributions are typically intractable.
\nsmc addresses this difficulty by requiring only approximate, \emph{properly weighted}, samples
from the proposal distribution. This enables us to use a second layer of \smc to simulate approximately from the proposal. The proper weighting condition ensures the validity of \nsmc,
thus providing a generalisation of the family of \smc methods. This paper extends preliminary work \citep{naessethLS2015nested} with the ability to handle more expressive models, more informative central limit theorems and convergence proofs, as well as new experiments.

% ======================================================================
%                             Related Work
% ======================================================================
% ======================================================================
%                       RELATED WORK
% ======================================================================
\subsection*{Related work}\label{sec:rw}
There has been much recent interest in using Monte Carlo methods as nested procedures of other Monte Carlo algorithms. The $\operatorname{SMC}^2$ and $\operatorname{IS}^2$ algorithms by \citet{chopinJP2013smc2} and \citet{tranSPK2013importance}, respectively, are algorithms for learning static parameters as well as latent variable(s). In these methods one \smc/\is method for the parameters is coupled with another for the latent variables. \citet{chenSOL2011decentralized} and \citet{johansenWD2012exact} on the other hand addresses the state inference problem by splitting $x_t$ into two components and run coupled \smc samplers for these. These methods solve different problems and the ``internal'' \smc samplers are constructed differently, for approximate marginalization instead of simulation.

By viewing the state inference problem as a sequential problem in the components of $x_t$ we can make use of the method for general graphical models by \citet{naessethLS2014}. This method is combined with the island particle filter \citep{vergeDDM2013on}, and studied more closely by \citet{beskosCJKZ2014a} under the name space-time particle filter (\stpf). The \stpf does not generate an approximation of the fully adapted \smc. Another key distinction is that in \stpf each particle in the ``outer'' \smc sampler corresponds to a complete particle system, whereas for \nsmc it will correspond to different hypotheses about the latent state $x_t$ as in standard \smc. This leads to lower communication costs and better memory efficiency in \eg distributed implementations. We have also found that \nsmc typically outperforms \stpf, even when run on a single machine with matched computing times.

The method proposed by \citet{jaoua2013bayesian} can be viewed as a special case of \nsmc when the nested procedure to generate samples is given by \is with the proposal being the transition probability. Independent resampling PF (IR-PF) introduced in \citet{lamberti2016independent} generates samples in the same way as \nsmc with \is, instead of \smc, as the nested procedure. However, IR-PF uses a different weighting that requires both the outer and the inner number of particles to tend to infinity for consistency. Furthermore, we provide results in the supplementary material that show \nsmc significantly outperforming IR-PF on an example studied in \citet{lamberti2016independent}.

There are other \smc-related methods that have been introduced to tackle high-dimensional problems, see \eg the so-called block PF studied by \citet{rebeschini2015}, the location particle smoother by \citet{briggsDM2013data}, and various methods reviewed in \citet{djuric2013particle}. These methods are, however, all inconsistent because they are based on approximations that result in systematic errors.

The concept of proper weighting (or random weights) is not new and has been used in the so-called random weights particle filter \citep{fearnheadPRS2010random}. They require exact samples from a proposal $q_t$ but use a nested Monte Carlo method to unbiasedly estimate the importance weights $w_t$. In \citet{martino2016weighting} the authors study proper weighting as a means to perform \emph{partial resampling}, \ie only resample a subset of the particles at each time. The authors introduce the concept of ``unnormalized'' proper weighting, which is essentially the same as proper weighting that was introduced and used to motivate \nsmc in \citet{naessethLS2015nested}. Furthermore, \citet{stern2015} uses proper weighting and \nsmc to solve an inference problem within statistical historical linguistics.

Another approach to solve the sequential inference problem is the sequential Markov chain Monte Carlo class of methods \citep{yang2013}. It was shown by \citet{septier2016} that the optimal sequential \mcmc algorithm actually is equivalent to the fully adapted \smc.

% ======================================================================
%                        Motivating examples
% ======================================================================
% ======================================================================
%                       MOTIVATING EXAMPLES
% ======================================================================
\section{Sequential probabilistic models}\label{sec:models}
In statistics, data science and machine learning, probabilistic modeling and Bayesian inference are essential tools to finding underlying patterns and unobserved quantities of interest. To illustrate the nested \smc sampler we will make use of two general classes of sequential probabilistic models, the so-called \emph{Markov random field} (\mrf) and the \emph{state space model} (\ssm). Sequential probabilistic models are in general built up of a sequence of (probabilistic) models that share common random variables and structure. These models will serve to illustrate the usefullness and wide applicability of the method we propose. We are interested in the type of sequential models where the latent variables are fairly high-dimensional. In subsequent sections we will also show explicitly how we can make use of structure between the (latent) random variables to design an efficient \smc sampler that lets us scale to much higher dimensions than possible with standard \smc methods, usually by up to 1--2 orders of magnitude. Note also that the \nsmc is by no means restricted to the classes of models we illustrate in this section, rather it can in principle be applied to any sequence of distributions we would like to approximate. We will refer to this sequence of distributions of interest as the \emph{target distributions}.

%Note that typically we will assume that the latent variables $x$ and the observed variables $y$ are high-dimensional and admit densities with respect to some dominating measures $\myd x$ and $\myd y$, respectively.

% ======================================================================
%                      MARKOV RANDOM FIELD
% ======================================================================
\subsection{Markov random fields}\label{sec:mrf}
The Markov random field is a type of undirected probabilistic graphical model \citep{jordan2004graphical}. The \mrf is typically not represented as a sequence of distributions (or models), but it has previously been shown \citep{hamze2005hot,everitt2012bayesian,NaessethLS:2014IT,naessethLS2014,naessethLS2015nested,naessethLS2015ws,lindstenjnksab2014} that it can be very useful to artificially introduce a sequence to simplify the inference problem. Furthermore, it is also possible to postulate the model as an \mrf that increases with ``time'', useful in \eg climate science \citep{fuBLS2012drought,naessethLS2015nested}. In the exposition below we will first for simplicity assume that we have an \mrf that is of fixed dimension, \ie the latent variable $x=(x_1,\ldots,x_{n_x})$ is a finite-dimensional multivariate random variable. The conditional independencies of an \mrf are described by the structure of the graph $G = \{\Ve,\Ed\}$, where $\Ve = \{1,\ldots,n_x\}$ is the vertex set and $\Ed = \{(i,j): (i,j) \in \Ve \times \Ve,~\exists~\text{edge between vertex}~i~\text{and}~j\}$ is the edge set. Given $G$ we can define a joint probability density function for $x$ that incorporates this structure as
\begin{align}
\pi(x) = \frac{1}{Z} \prod_{i\in \Ve} \phi(x_i,y_i) \prod_{(i,j) \in \Ed} \psi(x_i,x_j),
\label{eq:mrf}
\end{align}
where $y=(y_1,\ldots,y_{n_x})$ is the observed variable and $\phi,\psi$ are called observation and interaction potentials, respectively. The normalization constant that ensures that $\pi(\cdot)$ integrates to one is given by
\[
Z \eqdef \int \prod_{i\in \Ve} \phi(x_i,y_i) \prod_{(i,j) \in \Ed} \psi(x_i,x_j) \myd x.
\] 
Note that \eqref{eq:mrf} is usually referred to as a pairwise \mrf in the literature due to $\pi(\cdot)$ factorising into potentials that only depend on pairs of components of the random variable $x$. For clarity we restrict ourselves to this type, however the method we propose in this paper can be applied to more general types of graphs, see \eg \citet{naessethLS2014} for ideas on how to extend \smc inference to non-pairwise \mrf{s}. 

Now, the sequential \mrf is obtained if we consider a random variable $x_{1:t} = (x_1,\ldots,x_t)$, for some $t =1,\ldots,T$, that factorises according to
\begin{align}
\pi_t(x_{1:t}) = \frac{1}{Z_t} \gamma_t(x_{1:t}) \eqdef \frac{1}{Z_t} \gamma_t(x_{1:t-1}) \prod_{i\in \Ve} \phi(x_{t,i},y_{t,i}) \rho(x_{t-1,i},x_{t,i}) \prod_{(i,j) \in \Ed} \psi(x_{t,i},x_{t,j}),
\label{eq:seq:mrf}
\end{align}
where $G = \{\Ve,\Ed\}$ again encodes the structure of the graphical model and $\rho(\cdot)$ is a new type of interaction potential that links $x_{t-1}$ to $x_t$. Furthermore, the normalisation constant is given by $Z_t \eqdef \int \gamma_t(x_{1:t}) \myd x_{1:t}$. %and the initial distribution we get by setting $\gamma_0$ as follows
%\begin{align*}
%\gamma_0(x_{0}) = \prod_{i\in \Ve} \phi(x_{0,i},y_{0,i}) \prod_{(i,j) \in \Ed} \psi(x_{0,i},x_{0,j}).
%\end{align*}
We illustrate a typical example of a sequential \mrf in Figure~\ref{fig:seq:mrf}. It can amongst other things be used to model spatio-temporal phenomena, it was \eg used by \citet{naessethLS2015nested} to detect drought based on annual average precipitation rates collected from various sites in North America and Africa over the last century.
\begin{figure}[h]
\centering
\input{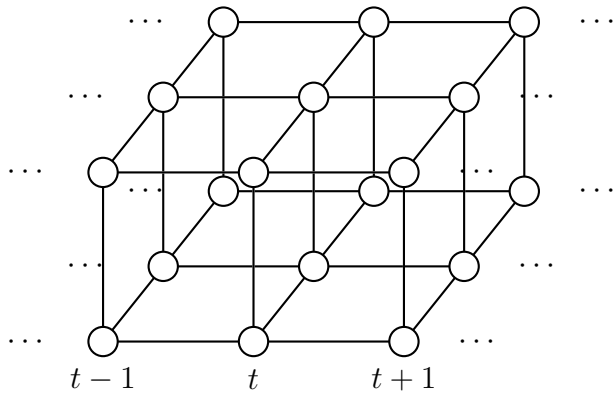}
\caption{Illustration of a sequential \mrf where $G$ is given by $2 \times 3$ grid with nearest neighbour interaction.}\label{fig:seq:mrf}
\end{figure}
We would like to remark on one peculiarity that arises when the sequential \mrf is used to model a spatio-temporal process. Consider $\pi_t(\cdot)$ without measurements as a prior on a spatio-temporal model, \ie the observation potentials $\phi$ in \eqref{eq:seq:mrf} do not depend on $y_t$. In this case we get that the marginals for $t < T$ change depending on the value of $T$, \ie in general ${\pi_t(x_{1:t}) \neq \pi_T(x_{1:t}) = \int \pi_T(x_{1:T}) \myd x_{t+1:T}}$.  Typically we would expect that \emph{a priori} what happens for a dynamical process at time $t$ should not be affected by the length of time-series we consider. The next class of models we consider can introduce dependencies in both time and space without giving rise to this counter-intuitive result.

% ======================================================================
%                       STATE SPACE MODELS
% ======================================================================
%\subsection{State space models}

% ======================================================================
%                       SPATIO-TEMPORAL SSMs
% ======================================================================
\subsection{Spatio-temporal state space models}
Before we move on to define the spatio-temporal state space model (\stssm), we will briefly review \ssm{s}, a comprehensive and important model type commonly used for studying dynamical systems. For a more detailed account, and with pointers to the wide range of applications, we refer the readers to \eg \citet{cappe2005,douc2014nonlinear,shumway2010time}. 

In state space models the sequential structure typically enters as a known, or postulated, dynamics on the unobserved latent state $x_t$ that is then partially observed through the measurements $y_t$. A common definition for \ssm{s} is through its functional form
\begin{subequations}
\begin{align}
x_t &= a(x_{t-1},v_t), &v_t \sim p_v (\cdot),\label{eq:ssm:x:functional}\\
y_t &= c(x_t,e_t), &e_t \sim p_e (\cdot),\label{eq:ssm:y:functional}
\end{align}
\label{eq:ssm:functional}
\end{subequations}
where $v_t$ and $e_t$, often called process and measurement noise, respectively, are random variables with some given distributions $p_v(\cdot), p_e(\cdot)$. Furthermore, we have that the initial state $x_1$ is a random variable with some initial distribution $\mu(\cdot)$. For simplicity we will assume that both $a(x_{t-1},\cdot) : \reals^{n_x} \to \reals^{n_x}$ and $c(x_t,\cdot) : \reals^{n_y} \to \reals^{n_y}$ are bijective and continuously differentiable. Then by the transformation theorem we can equivalently express \eqref{eq:ssm:functional} through the corresponding probability density functions (\pdf)
\begin{subequations}
\begin{align}
x_t | x_{t-1} &\sim f(x_t|x_{t-1}),\label{eq:ssm:x:density}\\
y_t | x_t &\sim g(y_t|x_t),\label{eq:ssm:y:density}
\end{align}
\label{eq:ssm:density}
\end{subequations}
and we define the sequential probabilistic model (or target distribution) as follows
\begin{align}
\pi_t(x_{1:t}) \eqdef \frac{\gamma_t(x_{1:t})}{Z_t} = \frac{1}{Z_t} \mu(x_1) g(y_1|x_1)\prod_{s=2}^t f(x_s|x_{s-1}) g(y_s|x_s).
\label{eq:ssm}
\end{align}

%The probability density functions $f$ and  $g$ are given as follows
%\begin{align}
%f(x_t|x_{t-1}) &= p_v(a^{-1}(x_{t-1},x_t)) \left|\det{\frac{\partial a^{-1}(x_{t-1},x_t)}{\partial x_t}}\right|,\nonumber\\
%g(y_t | x_t) &= p_e(c^{-1}(x_{t},y_t)) \left|\det{\frac{\partial c^{-1}(x_{t},y_t)}{\partial y_t}}\right|,
%\label{eq:xy:density}
%\end{align}
%where $a^{-1}(x_{t-1},\cdot), c^{-1}(x_t,\cdot)$ are the inverses of $a(x_{t-1},\cdot), c(x_t,\cdot)$, respectively, and $\frac{\partial a^{-1}}{\partial x_t}, \frac{\partial c^{-1}}{\partial y_t}$ denotes their corresponding Jacobians. 

We will assume that $g(y_t|x_t)$ is available and can be evaluated pointwise. % However, we do not make any such assumptions regarding the transition density $f(x_t|x_{t-1})$. 
This condition is often satisfied in practical applications.

A typical assumption when using the \ssm to model spatio-temporal systems is to introduce the spatial dependency only between time steps $t-1$ and $t$, see \eg the paper by \citet{wikle2010general}. This can be achieved by defining a model $a(\cdot)$ such that the product of the induced distributions $f(x_t|x_{t-1})g(y_t|x_t)$, conditionally on $x_{t-1}$, completely factorize over the components of $x_t$, see also \citep{rebeschini2015} where \smc applied to such a model is studied. Here we will study the case where we introduce spatial dependencies within each time step through the disturbance term $v_t$. We define the \stssm as a combination of the functional and \pdf representation of an \ssm where the distribution for $v_t$ is given by an \mrf as in \eqref{eq:mrf}
\begin{subequations}
\begin{align}
\left(
\begin{array}{c}
x_{t,1}\\
\vdots \\
x_{t,n_x}
\end{array}
\right)
&= 
\left(
\begin{array}{c}
a_1(x_{t-1},v_{t,1}) \\
\vdots \\
a_{n_x}(x_{t-1},v_{t,n_x})
\end{array}
\right),
&v_t \sim \frac{1}{Z_v} \prod_{i\in \Ve} \phi(v_{t,i}) \prod_{(i,j) \in \Ed} \psi(v_{t,i},v_{t,j}),\label{eq:stssm:x:functional}\\
y_t | x_t &\sim g(y_t|x_t).\label{eq:stssm:y:density}
\end{align}
\label{eq:stssm}
\end{subequations}

We make no assumptions on local dependencies between $x_t$ and $x_{t-1}$, however, to keep it simple we will assume that the graph $G = \{\Ve,\Ed\}$ describing the distribution for $v_t$ does not depend on time $t$.
%The \stssm is achieved for the above formulation, \ie \eqref{eq:stssm}, when $p_v(v_t)$ is modelled by a \mrf as in \eqref{eq:mrf}. 
Furthermore, we will in this paper mainly consider models where dependencies between components in $v_t$ are ``few'', \eg the \mrf is sparse with few elements in $\Ed$, and where components of $y_t$ in $g$ only depends on subsets of $x_t$. To illustrate the dependency structure in an \stssm we propose a combination of the traditional undirected graph for the \mrf and the directed acyclic graph for the \ssm, see Figure~\ref{fig:stssm}.
\begin{figure}[h]
\centering
\tikzstyle{edge} = [-]
\tikzstyle{edge2} = [->,very thick,>=latex]
\tikzstyle{edge3} = [->]
\tikzstyle{arrw} = [very thick,shorten <=2pt,shorten >=2pt]
\tikzstyle{var} = [draw,circle,inner sep=0,minimum width=0.5cm]
\tikzstyle{initvar} = [draw,circle,inner sep=0,minimum width=0.9cm]
\tikzstyle{obs} = [draw,circle,inner sep=0,minimum width=0.5cm, fill=black!20]
  \begin{tikzpicture}[>=stealth,node distance=0.6cm]
    \begin{scope}
      % Draw x-nodes and observations
      \foreach \x in {0,1,2,3,4,5} {
        \foreach \y in {0,1,2,3} {
          \pgfmathtruncatemacro\xend{\x+1}
          \pgfmathtruncatemacro\yend{4-\y}
          \node at (1.9*\x,\y) (x\x\y) [var] {$x_{\xend,\yend}$};
          \node at (1.9*\x+0.9,\y-0.15) (y\x\y) [obs] {};
        }
      }
      
      % Draw bounding boxes
      \foreach \x in {0,1,2,3,4,5} {
        \pgfmathtruncatemacro\xend{\x+1}
        \node[draw,very thick,rectangle,rounded corners=3mm,minimum width=0.6cm,fit=(x\x0) (x\x3)] (x\x){};
      }
      
      %\node at (-1.9,1.5) (xinit) [initvar] {$x_{1}$};
      %\draw[edge2] (xinit) -- (x0);
      
      % Draw horizontal edges
      \foreach \x in {0,1,2,3,4} {
        \pgfmathtruncatemacro\xend{\x+1}
          \draw[edge2] (x\x) -> (x\xend);
      }
      % Draw diagonal edges
      \foreach \x in {0,1,2,3,4,5} {
        \pgfmathtruncatemacro\xend{\x+1}
        \foreach \y in {0,1,2,3} {
          \draw[edge3] (x\x\y) -- (y\x\y);
        }
      }
      % Draw vertical edges
      \foreach \x in {0,1,2,3,4,5} {
        \foreach \y in {0,1,2} {
          \pgfmathtruncatemacro\yend{\y+1}
          \draw[edge] (x\x\y) -- (x\x\yend) {};
        }
      }
      \node at (6*1.9,1.5) (dots) [circle] {$\cdots$};
      \draw[edge2] (x5) -> (dots) {};
    \end{scope}
     %\draw [blue] (current bounding box.south west) rectangle (current bounding box.north east);
  \end{tikzpicture}
\caption{Illustration of a spatio-temporal state space model with $n_x=4$, one conditionally independent observation per component in $x_t$, and the \mrf for $v_t$ is given by a chain. Grey circles illustrate the observations $y_t$.}\label{fig:stssm}
\end{figure}
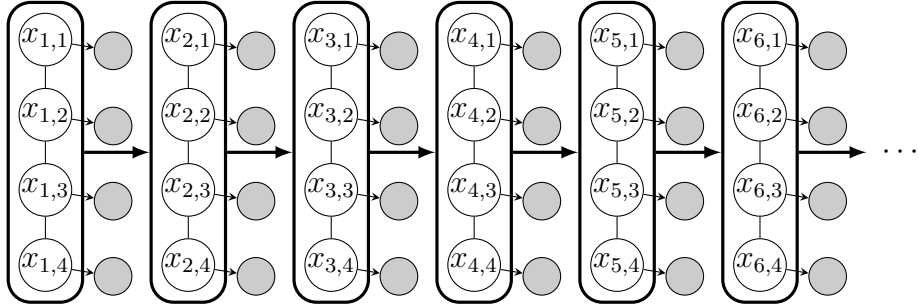
This allows us to model more complex dynamical processes than \citet{naessethLS2015nested} who assumed that $f(x_t|x_{t-1})g(y_t|x_t)$ factorized with only local dependencies between components of $x_t$. Furthermore, we can clearly see that the peculiarity discussed in Section~\ref{sec:mrf} is not present in this model; the marginal of the prior does not change with $T$ as expected.

% ======================================================================
%                         Background/Nested IS
% ======================================================================
%\input{sec-background}

%\newpage
% ======================================================================
%                         Method & Theory
% ======================================================================
\section{Nested Sequential Monte Carlo Methods}\label{sec:method}
Inference in sequential probabilistic models essentially boils down to computing the target distribution $\pi_t(x_{1:t})$ for $t=1,2,\ldots$; typically an intractable problem with no analytical or numerically efficient solution. This means that we have to resort to approximations. In this paper we focus on one particular succesful solution to the problem, the so called sequential Monte Carlo family of algorithms first introduced in the papers by \citet{GordonSS:1993,stewart1992,kitagawa1996monte}. 

The basic idea with \smc is to move a set of weighted samples (particles) $\{(x_{1:t-1}^i, w_{t-1}^i)\}_{i=1}^N$ approximating $\pi_{t-1}$, to a new set of particles $\{(x_{1:t}^i, w_{t}^i)\}_{i=1}^N$ which approximates $\pi_t$. These samples define an empirical approximation of the target distribution
\begin{align}
\pi_t^N(\myd x_{1:t}) \eqdef \sum_{i=1}^N \frac{w_t^i}{\sum_\ell w_t^\ell} \delta_{x_{1:t}^i}(\myd x_{1:t}),
\end{align}
where $\delta_{x}(\myd x)$ is a Dirac measure at $x$. In the next section we will detail an especially efficient way of moving the particles, known as fully adapted \smc \citep{pittS1999filtering}, ensuring that all normalized weights are equal to $\frac{1}{N}$.

% ======================================================================
%                             FA SMC
% ======================================================================
\subsection{Fully Adapted Sequential Monte Carlo}\label{sec:fapf}
The procedure to move the particles and their weights from time $t-1$ to $t$ in any \smc sampler is typically done in a three-stage approach. The first, \emph{resampling}, stochastically chooses $N$ particles at time $t-1$ that seem promising, discarding low-weighted ones. The second stage, \emph{propagation}, generates new samples for time $t$ conditioned on the resampled particles. The final stage, \emph{weighting}, corrects for the discrepancy between the target distribution and the \emph{proposal}, \ie the instrumental distribution used in the propagation step.

Fully adapted \smc \citep{pittS1999filtering} makes specific choices on the resampling weights, $\nu_{t-1}$, and the proposal, $q_t(x_t|x_{1:t-1})$, such that all the importance weights $w_t$ are equal. By introducing \emph{ancestor indices} $a_{t-1}^i \in \{1,\ldots,N\}$, we can describe the resampling step by simulating $N$ times \iid from
\begin{align}
\Prb(a_{t-1}^i = j) = \frac{\nu_{t-1}^j}{\sum_{\ell=1}^M \nu_{t-1}^\ell}, \quad \nu_{t-1}^j \eqdef \int \frac{\gamma_t\left((x_{1:t-1}^j,x_t)\right)}{\gamma_{t-1}(x_{1:t-1}^j)} \myd x_t.
\end{align}
Propagation then follows by simulating $x_t^i$ conditional on $x_{1:t-1}^{a_{t-1}^i}$, for $i=1,\ldots,N$, as follows
\begin{align}
x_t^i | x_{1:t-1}^{a_{t-1}^i} &\sim q_t(x_t|x_{1:t-1}^{a_{t-1}^i}) \eqdef \pi_t(x_t |x_{1:t-1}^{a_{t-1}^i}) = \frac{\pi_t((x_{1:t-1}^{a_{t-1}^i},x_t))}{\pi_t(x_{1:t-1}^{a_{t-1}^i})}, \\
x_{1:t}^i &= \left(x_{1:t-1}^{a_{t-1}^i},x_t^i\right).\nonumber
\end{align}
This proposal is sometimes referred to as the (locally) \emph{optimal proposal} because it minimizes incremental variances in the importance weights $w_t^i$. Weighting is easy since all weights are equal, \ie the unnormalized weights are all set to $w_t^i = 1$. The fully adapted \smc sampler  in fact corresponds to a locally optimal choice of both resampling weights and proposal with an incremental variance in the importance weights $w_t^i$ that is zero.

Note that in most cases it is impossible to implement this algorithm exactly, since we can not calculate $\nu_{t-1}$ and/or simulate from $q_t$. Nested \smc solves this by \emph{requiring only approximate resampling weights and approximate samples from $q_t$}, in the sense that is formalized in Section~\ref{sec:nsmc}. However, we will start by detailing some specific cases when we can efficiently implement exact fully adapted \smc. These cases are of interest in themselves, however, here we will use them to build intuition for how the approximations in \nsmc are constructed.

%
%\subsubsection{Forward Filtering--Backward Simulation}
\subsection{Forward Filtering--Backward Simulation}\label{sec:ffbs}
The problems we need to solve are those of computing $\nu_{t-1}$ and simulating from $q_t$ efficiently, \ie in such a way that the computational complexity is controlled. %If we for a minute disregard computational complexity 
There are at least two important special cases when we can use fully adapted \smc. The first is if the state space $\setX$ is discrete and finite, \ie $x_t \in \{1,\ldots,S\}^{\otimes n_x}, \forall t$. Even though exact algorithms are known in this case \citep{cappe2005} the computational complexity typically scales quadratically with the cardinality of $x_t$, thus \smc methods can still be of interest \citep{fearnhead2003line,NaessethLS:2014IT,naessethLS2015nested}. The second case is if $\frac{\gamma_t(x_{1:t})}{\gamma_{t-1}(x_{1:t-1})}$ is an unnormalized Gaussian distribution, \eg in the \stssm this would correspond to 
\begin{align*}
x_t &= a(x_{t-1})+v_t, \quad v_t \sim \text{Gaussian \mrf},\\
y_t | x_t &\sim \N(y_t; C x_t, R),
\end{align*}
for some matrix $C$, covariance matrix $R$, and an \mrf in the components of $v_t$ where all pair-wise potentials are Gaussian.

Now, even though in principle the fully adapted \smc is available these special cases, the computational complexity can be prohibitive. In fact in general it is of the order of $\Ordo(S^{n_x})$ and $\Ordo(n_x^3)$ for the finite state space and Gaussian case, respectively. However, when there are local dependencies it is possible to make use of an underlying chain (or tree) structure, as proposed by \citet{NaessethLS:2014IT} for the finite state space case, to make efficient implementations with only $\Ordo(S^2 n_x)$ and $\Ordo(n_x)$ complexity, respectively. This approach makes use of forward filtering--backward simulation (sampling), from \citet{carter1994gibbs,fruhwirth1994data}, on the \emph{components} of $x_t$ to compute $\nu_{t-1}$ and sample $q_t$ exactly. Let us as an example consider the above \stssm with $C=I$ and $R=I$ and the Gaussian \mrf given by
\begin{align*}
p_v(v_t) = \frac{1}{Z_v} \exp{\left\{-\frac{\tau}{2} \sum_{d=1}^{n_x} v_{t,d}^2 -\frac{\lambda}{2} \sum_{d=2}^{n_x}(v_{t,d}-v_{t,d-1})^2\right\}},
\end{align*}
for some positive constants $\tau$ and $\lambda$. Then straightforward computations gives the proposal and resampling weights
\begin{align*}
q_t(x_t|x_{1:t-1}) &= \frac{1}{\nu_{t-1}}\frac{\gamma_t(x_{1:t})}{\gamma_{t-1}(x_{1:t-1})} = \frac{1}{\nu_{t-1}}f(x_t|x_{t-1})g(y_t|x_t),\\
\nu_{t-1} &= \int f(x_t|x_{t-1})g(y_t|x_t) \myd x_t. 
\end{align*}
However, an equivalent way to simulate from this distribution and calculate $\nu_{t-1}$ is given below
\begin{align*}
x_t &= a(x_{t-1})+v_t', \quad v_t' \sim \frac{1}{\nu_{t-1}} g(y_t|a(x_{t-1})+v_t) p_v(v_t),\\
\nu_{t-1} &= \int g(y_t|a(x_{t-1})+v_t) p_v(v_t) \myd v_t.
\end{align*}
Due to the structure in $p_v(\cdot)$ and $g(y_t|x_t)$ we can see that the distribution to sample from corresponds to a Gaussian \mrf with a chain-structure in the $v_{t,d}$'s (\cf Figure~\ref{fig:stssm})
\begin{align}
&p(v_{t,1:n_x}|y_t,x_{t-1}) = \frac{g(y_t|a(x_{t-1})+v_t) p_v(v_t)}{\prod_{d=1}^{n_x} p(y_{t,d}|y_{t,1:d-1},x_{t-1})} \nonumber\\
&\quad\propto\frac{1}{Z_v}\exp{\left\{ -\frac{1}{2}\sum_{d=1}^{n_x} \left[(y_{t,i}-a_d(x_{t-1})-v_{t,d})^2 +\tau v_{t,d}^2\right] -\frac{\lambda}{2} \sum_{d=2}^{n_x}(v_{t,d}-v_{t,d-1})^2\right\}}. \label{eq:ex:optprop}
\end{align}
Because of this structure we can efficiently compute the normalization constant of \eqref{eq:ex:optprop} by means of ``forward'' filtering, keeping track of the incremental contributions to $\nu_{t-1}$, $p(y_{t,d}|y_{t,1:d-1},x_{t-1}), ~d=1,\ldots,n_x$. Sampling the distribution is then done by an explicit ``backward'' pass, simulating $v_{t,d}' \sim p(v_{t,d} | v_{t,d+1:n_x}', y_{t,d:n_x})$, $d=n_x,n_x-1,\ldots,1$. We provide an illustration of the process in Figure~\ref{fig:ffbs}. See also \citet{NaessethLS:2014IT} for an example of how this is done in practice for a discrete state space.

The main idea behind nested \smc is to emulate this behavior for arbitrary sequential probabilistic models. Because computing $\nu_{t-1}$ and simulating from $q_t$ exactly is intractable in general we propose to run an \smc-based forward filtering--backward simulation \citep{GodsillDW:2004,LindstenS:2013} method on the components of $x_t$ (or $v_t$) to approximate $\nu_{t-1}$ and draws from $q_t$.
\begin{figure}[h]
\hspace{-5mm}
\centering
\begin{subfigure}[b]{0.24\textwidth}
\resizebox{1.0\textwidth}{!}{
\tikzstyle{edge} = [-]
\tikzstyle{edge2} = [->,very thick,>=latex]
\tikzstyle{edge3} = [->]
\tikzstyle{arrw} = [very thick,shorten <=2pt,shorten >=2pt]
\tikzstyle{var} = [draw,circle,inner sep=0,minimum width=0.7cm]
\tikzstyle{cvar} = [draw,circle,inner sep=0,minimum width=0.7cm, fill=black!20]
\tikzstyle{initvar} = [draw,circle,inner sep=0,minimum width=0.9cm]
\tikzstyle{obs} = [draw,circle,inner sep=0,minimum width=0.5cm, fill=black!20]
  \begin{tikzpicture}[>=stealth,node distance=0.6cm]
    \begin{scope}
      % Draw x-nodes and observations
      \foreach \x in {1} {
        \foreach \y in {0,1,2,3} {
          \pgfmathtruncatemacro\xend{\x+1}
          \pgfmathtruncatemacro\yend{4-\y}
          \node at (1.9*\x,\y) (x\x\y) [cvar] {};
        }
      }
      \foreach \x in {2} {
        \foreach \y in {3} {
          \pgfmathtruncatemacro\xend{\x+1}
          \pgfmathtruncatemacro\yend{4-\y}
          \node at (1.9*\x,\y) (x\x\y) [var] {$v_{t,\yend}$};
          \node at (1.9*\x+0.9,\y-0.15) (y\x\y) [obs] {};
        }
      }
      
      % Draw bounding boxes
      \foreach \x in {1} {
        \pgfmathtruncatemacro\xend{\x+1}
        \node[draw,very thick,rectangle,rounded corners=3mm,minimum width=0.6cm,fit=(x\x0) (x\x3), label=above:{$x_{t-1}$}] (x\x){};
      }
      \foreach \x in {2} {
        \pgfmathtruncatemacro\xend{\x+1}
        \node[draw,very thick,rectangle,rounded corners=3mm,minimum width=0.6cm,fit=(x\x0) (x\x3), label=above:{$v_{t,1}$}] (x\x){};
      }
      %\node at (-1.9,1.5) (xinit) [initvar] {$x_{1}$};
      %\draw[edge2] (xinit) -- (x0);
      
      % Draw horizontal edges
      \foreach \x in {1} {
        \pgfmathtruncatemacro\xend{\x+1}
          \draw[edge2] (x\x) -> (x\xend);
      }
      % Draw diagonal edges
      \foreach \x in {2} {
        \pgfmathtruncatemacro\xend{\x+1}
        \foreach \y in {3} {
          \draw[edge3] (x\x\y) -- (y\x\y);
        }
      }
      % Draw vertical edges
      \foreach \x in {1} {
        \foreach \y in {0,1,2} {
          \pgfmathtruncatemacro\yend{\y+1}
          \draw[edge] (x\x\y) -- (x\x\yend) {};
        }
      }
      \node at (0.5,1.5) (dots) [circle] {$\cdots$};
    \end{scope}
     %\draw [blue] (current bounding box.south west) rectangle (current bounding box.north east);
  \end{tikzpicture}
}
\end{subfigure}
\begin{subfigure}[b]{0.24\textwidth}
\resizebox{1.0\textwidth}{!}{
\tikzstyle{edge} = [-]
\tikzstyle{edge2} = [->,very thick,>=latex]
\tikzstyle{edge3} = [->]
\tikzstyle{arrw} = [very thick,shorten <=2pt,shorten >=2pt]
\tikzstyle{var} = [draw,circle,inner sep=0,minimum width=0.7cm]
\tikzstyle{cvar} = [draw,circle,inner sep=0,minimum width=0.7cm, fill=black!20]
\tikzstyle{initvar} = [draw,circle,inner sep=0,minimum width=0.9cm]
\tikzstyle{obs} = [draw,circle,inner sep=0,minimum width=0.5cm, fill=black!20]
\begin{tikzpicture}[>=stealth,node distance=0.6cm]
    \begin{scope}
      % Draw x-nodes and observations
      \foreach \x in {1} {
        \foreach \y in {0,1,2,3} {
          \pgfmathtruncatemacro\xend{\x+1}
          \pgfmathtruncatemacro\yend{4-\y}
          \node at (1.9*\x,\y) (x\x\y) [cvar] {};
        }
      }
      \foreach \x in {2} {
        \foreach \y in {2,3} {
          \pgfmathtruncatemacro\xend{\x+1}
          \pgfmathtruncatemacro\yend{4-\y}
          \node at (1.9*\x,\y) (x\x\y) [var] {$v_{t,\yend}$};
          \node at (1.9*\x+0.9,\y-0.15) (y\x\y) [obs] {};
        }
      }
      
      % Draw bounding boxes
      \foreach \x in {1} {
        \pgfmathtruncatemacro\xend{\x+1}
        \node[draw,very thick,rectangle,rounded corners=3mm,minimum width=0.6cm,fit=(x\x0) (x\x3), label=above:{$x_{t-1}$}] (x\x){};
      }
      \foreach \x in {2} {
        \pgfmathtruncatemacro\xend{\x+1}
        \node[draw,very thick,rectangle,rounded corners=3mm,minimum width=0.6cm,fit=(x\x0) (x\x3), label=above:{$v_{t,1:2}$}] (x\x){};
      }
      %\node at (-1.9,1.5) (xinit) [initvar] {$x_{1}$};
      %\draw[edge2] (xinit) -- (x0);
      
      % Draw horizontal edges
      \foreach \x in {1} {
        \pgfmathtruncatemacro\xend{\x+1}
          \draw[edge2] (x\x) -> (x\xend);
      }
      % Draw diagonal edges
      \foreach \x in {2} {
        \pgfmathtruncatemacro\xend{\x+1}
        \foreach \y in {2,3} {
          \draw[edge3] (x\x\y) -- (y\x\y);
        }
      }
      % Draw vertical edges
      \foreach \x in {1} {
        \foreach \y in {0,1,2} {
          \pgfmathtruncatemacro\yend{\y+1}
          \draw[edge] (x\x\y) -- (x\x\yend) {};
        }
      }
      % Draw vertical edges
      \foreach \x in {2} {
        \foreach \y in {2} {
          \pgfmathtruncatemacro\yend{\y+1}
          \draw[edge] (x\x\y) -- (x\x\yend) {};
        }
      }
      \node at (0.5,1.5) (dots) [circle] {$\cdots$};
    \end{scope}
     %\draw [blue] (current bounding box.south west) rectangle (current bounding box.north east);
  \end{tikzpicture}
}
\end{subfigure}
\begin{subfigure}[b]{0.24\textwidth}
\resizebox{1.0\textwidth}{!}{
\tikzstyle{edge} = [-]
\tikzstyle{edge2} = [->,very thick,>=latex]
\tikzstyle{edge3} = [->]
\tikzstyle{arrw} = [very thick,shorten <=2pt,shorten >=2pt]
\tikzstyle{var} = [draw,circle,inner sep=0,minimum width=0.7cm]
\tikzstyle{cvar} = [draw,circle,inner sep=0,minimum width=0.7cm, fill=black!20]
\tikzstyle{initvar} = [draw,circle,inner sep=0,minimum width=0.9cm]
\tikzstyle{obs} = [draw,circle,inner sep=0,minimum width=0.5cm, fill=black!20]
\begin{tikzpicture}[>=stealth,node distance=0.6cm]
    \begin{scope}
      % Draw x-nodes and observations
      \foreach \x in {1} {
        \foreach \y in {0,1,2,3} {
          \pgfmathtruncatemacro\xend{\x+1}
          \pgfmathtruncatemacro\yend{4-\y}
          \node at (1.9*\x,\y) (x\x\y) [cvar] {};
        }
      }
      \foreach \x in {2} {
        \foreach \y in {1,2,3} {
          \pgfmathtruncatemacro\xend{\x+1}
          \pgfmathtruncatemacro\yend{4-\y}
          \node at (1.9*\x,\y) (x\x\y) [var] {$v_{t,\yend}$};
          \node at (1.9*\x+0.9,\y-0.15) (y\x\y) [obs] {};
        }
      }
      
      % Draw bounding boxes
      \foreach \x in {1} {
        \pgfmathtruncatemacro\xend{\x+1}
        \node[draw,very thick,rectangle,rounded corners=3mm,minimum width=0.6cm,fit=(x\x0) (x\x3), label=above:{$x_{t-1}$}] (x\x){};
      }
      \foreach \x in {2} {
        \pgfmathtruncatemacro\xend{\x+1}
        \node[draw,very thick,rectangle,rounded corners=3mm,minimum width=0.6cm,fit=(x\x0) (x\x3), label=above:{$v_{t,1:3}$}] (x\x){};
      }
      %\node at (-1.9,1.5) (xinit) [initvar] {$x_{1}$};
      %\draw[edge2] (xinit) -- (x0);
      
      % Draw horizontal edges
      \foreach \x in {1} {
        \pgfmathtruncatemacro\xend{\x+1}
          \draw[edge2] (x\x) -> (x\xend);
      }
      % Draw diagonal edges
      \foreach \x in {2} {
        \pgfmathtruncatemacro\xend{\x+1}
        \foreach \y in {1,2,3} {
          \draw[edge3] (x\x\y) -- (y\x\y);
        }
      }
      % Draw vertical edges
      \foreach \x in {1} {
        \foreach \y in {0,1,2} {
          \pgfmathtruncatemacro\yend{\y+1}
          \draw[edge] (x\x\y) -- (x\x\yend) {};
        }
      }
      % Draw vertical edges
      \foreach \x in {2} {
        \foreach \y in {1,2} {
          \pgfmathtruncatemacro\yend{\y+1}
          \draw[edge] (x\x\y) -- (x\x\yend) {};
        }
      }
      \node at (0.5,1.5) (dots) [circle] {$\cdots$};
    \end{scope}
     %\draw [blue] (current bounding box.south west) rectangle (current bounding box.north east);
  \end{tikzpicture}
}
\end{subfigure}
\begin{subfigure}[b]{0.24\textwidth}
\resizebox{1.0\textwidth}{!}{
\tikzstyle{edge} = [-]
\tikzstyle{edge2} = [->,very thick,>=latex]
\tikzstyle{edge3} = [->]
\tikzstyle{arrw} = [very thick,shorten <=2pt,shorten >=2pt]
\tikzstyle{var} = [draw,circle,inner sep=0,minimum width=0.7cm]
\tikzstyle{cvar} = [draw,circle,inner sep=0,minimum width=0.7cm, fill=black!20]
\tikzstyle{initvar} = [draw,circle,inner sep=0,minimum width=0.9cm]
\tikzstyle{obs} = [draw,circle,inner sep=0,minimum width=0.5cm, fill=black!20]
\begin{tikzpicture}[>=stealth,node distance=0.6cm]
    \begin{scope}
      % Draw x-nodes and observations
      \foreach \x in {1} {
        \foreach \y in {0,1,2,3} {
          \pgfmathtruncatemacro\xend{\x+1}
          \pgfmathtruncatemacro\yend{4-\y}
          \node at (1.9*\x,\y) (x\x\y) [cvar] {};
        }
      }
      \foreach \x in {2} {
        \foreach \y in {0,1,2,3} {
          \pgfmathtruncatemacro\xend{\x+1}
          \pgfmathtruncatemacro\yend{4-\y}
          \node at (1.9*\x,\y) (x\x\y) [var] {$v_{t,\yend}$};
          \node at (1.9*\x+0.9,\y-0.15) (y\x\y) [obs] {};
        }
      }
      
      % Draw bounding boxes
      \foreach \x in {1} {
        \pgfmathtruncatemacro\xend{\x+1}
        \node[draw,very thick,rectangle,rounded corners=3mm,minimum width=0.6cm,fit=(x\x0) (x\x3), label=above:{$x_{t-1}$}] (x\x){};
      }
      \foreach \x in {2} {
        \pgfmathtruncatemacro\xend{\x+1}
        \node[draw,very thick,rectangle,rounded corners=3mm,minimum width=0.6cm,fit=(x\x0) (x\x3), label=above:{$v_{t,1:n_x}$}] (x\x){};
      }
      %\node at (-1.9,1.5) (xinit) [initvar] {$x_{1}$};
      %\draw[edge2] (xinit) -- (x0);
      
      % Draw horizontal edges
      \foreach \x in {1} {
        \pgfmathtruncatemacro\xend{\x+1}
          \draw[edge2] (x\x) -> (x\xend);
      }
      % Draw diagonal edges
      \foreach \x in {2} {
        \pgfmathtruncatemacro\xend{\x+1}
        \foreach \y in {0,1,2,3} {
          \draw[edge3] (x\x\y) -- (y\x\y);
        }
      }
      % Draw vertical edges
      \foreach \x in {1,2} {
        \foreach \y in {0,1,2} {
          \pgfmathtruncatemacro\yend{\y+1}
          \draw[edge] (x\x\y) -- (x\x\yend) {};
        }
      }
      \node at (0.5,1.5) (dots) [circle] {$\cdots$};
    \end{scope}
     %\draw [blue] (current bounding box.south west) rectangle (current bounding box.north east);
  \end{tikzpicture}
}
\end{subfigure}
\caption{Illustration of forward filtering--backward sampling on $v_t$ as explained by \eqref{eq:ex:optprop}. Note that after the last step we simply set $x_t = a(x_{t-1}) + v_t'$.}\label{fig:ffbs}
\end{figure}
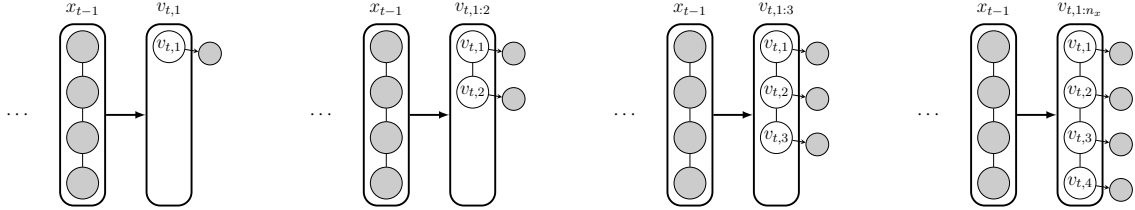

% ======================================================================
%                           Nested SMC
% ======================================================================
\subsection{Nested Sequential Monte Carlo}\label{sec:nsmc}
One way to think of the nested \smc family of methods is as an \emph{exact approximation} \citep{andrieuDH2010particle} of an \smc algorithm with resampling weights $\nu_{t-1}$ and proposal $q_t$ given as in the fully adapted \smc. Instead of exactly evaluating each $\nu_{t-1}^i$, we run a nested (or internal) \smc sampler with $M$ particles, for each $i$, on the components $x_{t,1:d}$ (or $v_{t,1:d}$) with the final target (for $d=n_x$) equal to $q_t(x_t|x_{1:t-1})$ to mimic the exact forward filtering procedure. The normalization constant estimates from these internal filters gives us \emph{unbiased} approximations of $\nu_{t-1}^i$ that we use to perform the resampling step. The resampling step not only selects the ancestors $x_{1:t-1}^{a_t^i}$, but we also resample the complete internal state, denoted by $u_{t-1}$, of the nested \smc samplers which will be used for the propagation step. Lastly we simulate $x_t^i$ by running a backward simulation procedure \citep{GodsillDW:2004,LindstenS:2013} using the resampled internal \smc sampler's $u_{t-1}^{a_t^i}$ to mimic the exact backward sampling described above. More formally, one step from iteration $t-1$ to $t$ of the \nsmc method proceeds as follows. 

% Resampling
Given an unweighted particle set $\{x_{1:t-1}^i\}_{i=1}^N$ ($w_{t-1}\equiv 1$), approximating $\pi_{t-1}$, we generate the internal states by simulating $u_{t-1}^i \sim \eta_{t-1}^M(u_{t-1}|x_{1:t-1}^i)$ (\cf forward filtering). Here $\eta_{t-1}^M$ denotes the joint distribution of all random variables generated by the internal \smc sampler. Then we extract an estimate of the resampling weights $\nu_{t-1}^i = \tau_t(u_{t-1}^i)$, where $\tau$ is a function such that
\begin{align}
\int \tau_t(u_{t-1}) \eta_{t-1}^M(u_{t-1}|x_{1:t-1}) \myd u_{t-1} &= \frac{\int \gamma_t((x_{1:t-1},x_t)) \myd x_t}{\gamma_{t-1}(x_{1:t-1})} = \nu_{t-1}, & \tau(u_{t-1}) \geq 0 ~\textrm{a.s.}\label{eq:nsmc:unbiasednu}
\end{align}
This is the normalization constant estimate at the final step of the internal \smc samplers, where the target is equal to $\frac{\gamma_t((x_{1:t-1}^i,x_t))}{\gamma_{t-1}(x_{1:t-1}^i)}$, and then \eqref{eq:nsmc:unbiasednu} is satisfied by known properties of \smc \citep[Proposition~7.4.1]{DelMoral:2004}. We now proceed to resample the internal \smc samplers, \ie generating ancestor variables $a_t^i$ such that
\begin{align}
\Prb(a_t^i = j) = \left\{\frac{\tau_t(u_{t-1}^j)}{\sum_\ell \tau_t(u_{t-1}^\ell)}\right\}_{j=1}^N,\label{eq:nsmc:resampling}
\end{align}
which concludes the resampling step.

% Propagation
Next, for propagation we generate samples $x_t^i \sim \kappa_t^M(x_t|u_{t-1}^{a_t^i})$ (\cf backward sampling), where $\kappa_t^M$ is a distribution satisfying the following condition
\begin{align}
\int \tau_t(u_{t-1}) \kappa_t^M(x_t|u_{t-1}) \eta_{t-1}^M(u_{t-1}|x_{1:t-1}) \myd u_{t-1} = \frac{\gamma_t(x_{1:t})}{\gamma_{t-1}(x_{1:t-1})}\label{eq:nsmc:prop}.
\end{align}
The distribution $\kappa_t^M$ can be realized by running backward simulation, however, a simple straightforward alternative that also satisfies \eqref{eq:nsmc:prop} can be to sample from the corresponding empirical distribution induced by the internal \smc sampler. We discuss the choice of $\eta_{t-1}^M, \kappa_t^M$ and $\tau_t$ further in the next section.

% Final
Finally, we set $x_{1:t}^i = (x_{1:t-1}^{a_t^i},x_t^i)$ and have thus obtained a new set of unweighted particles approximating $\pi_t$, \ie
\begin{align}
\pi_t^N(\myd x_{1:t}) \eqdef \frac{1}{N}\sum_{i=1}^N \delta_{x_{1:t}^i}(\myd x_{1:t}).
\end{align}

\begin{algorithm}
\caption{Nested Sequential Monte Carlo\hfill (all for $i=1,\ldots,N$)}\label{alg:nsmc}
\begin{algorithmic}[1]
\REQUIRE $\eta_{t-1}^M, \kappa_t^M, \tau_t$ that generate samples properly weighted for $\frac{\gamma_t(x_{1:t})}{\gamma_{t-1}(x_{1:t-1})}$
%\STATE $x_0^i \sim \mu(x_0)$
\FOR{$t=1$ to $T$}
\STATE Simulate $u_{t-1}^i \sim \eta_{t-1}^M(u_{t-1}|x_{1:t-1}^i)$
\STATE Draw $a_t^i$ with probability $\Prb(a_t^i = j) = \frac{\tau_t(u_{t-1}^j)}{\sum_\ell \tau_t(u_{t-1}^\ell)}$
\STATE Simulate $x_t^i \sim \kappa_t^M(x_t | u_{t-1}^{a_t^i})$
\STATE Set $x_{1:t}^i = (x_{1:t-1}^{a_t^i},x_t^i)$
\ENDFOR
\end{algorithmic}
\end{algorithm}

The two conditions on $\eta_{t-1}^M, \tau_t, \kappa_t^M$, \ie \eqref{eq:nsmc:unbiasednu} and \eqref{eq:nsmc:prop}, can in fact be replaced by the single condition that $(x_t^i,\tau_t(u_{t-1}^i))$ are \emph{properly weighted} for $\frac{\gamma_t(x_{1:t})}{\gamma_{t-1}(x_{1:t-1})}$.
\begin{definition}
We say that the (random) pair $(x_t,\tau_t(u_{t-1}))$ are properly weighted for the (unnormalized) distribution $\frac{\gamma_t(x_{1:t})}{\gamma_{t-1}(x_{1:t-1})}$ if $\tau_t(u_{t-1}) \geq 0 ~\textrm{a.s.}$ and for all measurable functions~$h$
\begin{align}
\E[h(x_t) \tau_t(u_{t-1})] = C\int \frac{\gamma_t(x_{1:t})}{\gamma_{t-1}(x_{1:t-1})} \myd x_t \int h(x_t) \pi_t(x_t|x_{1:t-1}) \myd x_t,
\end{align}
for some positive constant $C>0$ that is indepedent of the $x$'s and $u$'s.
\end{definition}
%\begin{remark}
%Note that this definition is more or less equivalent to the one used by \citet{naessethLS2015nested}, but we have made the dependence on $x_{1:t-1}$ explicit. The concept of proper weighting is also described by \citet{liu2001monte} in the context of importance sampling.
%\end{remark}
We provide a summary of the proposed method in Algorithm~\ref{alg:nsmc}. Although we here focus on approximating the fully adapted \smc sampler, the extension to arbitrary resampling weights and proposal is straightforward, see the supplementary material. Next we will illustrate how we can make use of \emph{nested} or internal \smc samplers to construct $\eta_{t-1}^M, \tau_t, \kappa_t^M$ that generate properly weighted samples.

% ======================================================================
%                     
% ======================================================================
\subsection{Constructing $\eta_{t-1}^M$, $\tau_t$ and $\kappa_t^M$}\label{sec:nested}
To construct $\eta_{t-1}^M$ we propose to run an \smc sampler targeting the components of $x_t$ (or $v_t$) one-by-one. This is done by choosing some sequence of (unnormalized) targets $p_d(x_{1:d})$ and proposals $r_d(x_d|x_{1:d-1})$ such that $p_{n_x}(\cdot) \propto \frac{\gamma_t(x_{1:t})}{\gamma_{t-1}(x_{1:t-1})}$. For notational convenience we supress the dependence on time $t$ in this section. We provide a summary in Algorithm~\ref{alg:smc}, in this case $u_{t-1} \eqdef \{x_d^{1:M}\}_{d=1}^{n_x} \bigcup \{a_d^{1:N}\}_{d=2}^{n_x}$.
\begin{algorithm}
\caption{Sequential Monte Carlo\hfill (all for $i=1,\ldots,M$)}\label{alg:smc}
\begin{algorithmic}[1]
\REQUIRE Unnormalized target distributions $p_d(x_{1:d})$, proposals $r_d(x_d|x_{1:d-1})$, and $M$
\STATE $x_1^i \sim r_1(x_1)$
\STATE Set $w_1^i = \frac{p_1(x_1^i)}{r_1(x_1^i)}$
\FOR{$d=2$ to $n_x$}
\STATE Draw $a_d^i$ with probability $\Prb(a_d^i = j) = \frac{w_{d-1}^j}{\sum_\ell w_{d-1}^\ell}$
\STATE Simulate $x_d^i \sim r_d(x_d | x_{d-1}^{a_d^i})$
\STATE Set $x_{1:d}^i = (x_{1:d-1}^{a_d^i},x_d^i)$
\STATE Set $w_d^i = \frac{p_d(x_{1:d}^i)}{p_{d-1}(x_{1:d-1}^{a_d^i}) r_d(x_d^i|x_{1:d-1}^{a_d^i})}$
\ENDFOR
\end{algorithmic}
\end{algorithm}

A first simple alternative to construct $\kappa_t^M$ can be to simply simulate directly from the empirical measure defined by the approximation in Algorithm~\ref{alg:smc}. Although this will be properly weighted it can introduce significant correlation between the samples. Instead we propose to make use of backward simulation \citep{GodsillDW:2004,LindstenS:2013} to construct a more efficient $\kappa_t^M$, see Algorithm~\ref{alg:bs}.
\begin{algorithm}
\caption{Backward Simulation}\label{alg:bs}
\begin{algorithmic}[1]
\REQUIRE $\{(x_{1:d}^i,w_d^i)\}_{i=1}^M, ~d=1,\ldots,n_x$ approximating $p_d(x_{1:d})$
\STATE Draw $b$ with probability $\Prb(b_{n_x} = j) = \frac{w_{n_x}^j}{\sum_\ell w_{n_x}^\ell}$
\STATE Set $x_{n_x} = x_{n_x}^{b}$
\FOR{$d=n_x-1$ to $1$}
\STATE Draw $b$ with probability $\Prb(b = j) \propto w_{d}^j \frac{p_{n_x}\left((x_{1:d}^j,x_{d+1:n_x})\right)}{p_d(x_{1:d}^j)}$
\STATE Set $x_{d:n_x} = (x_d^{b},x_{d+1:n_x})$
\ENDFOR
\end{algorithmic}
\end{algorithm}

Now, putting all this together we define the complete procedure in Definition~\ref{def:nested}.
\begin{definition}[\smc and BS]
\label{def:nested}
Let $\eta_{t-1}^M$, $\tau_t$, and $\kappa_t^M$ be defined as follows for some sequence $p_d(\cdot)$ such that $p_{n_x}(\cdot) \propto \frac{\gamma_t(x_{1:t})}{\gamma_{t-1}(x_{1:t-1})}$:
\begin{enumerate}
\item Simulate $u_{t-1}\sim\eta_{t-1}^M(u_{t-1}|x_{1:t-1})$ by running Algorithm~\ref{alg:smc}%, $u_{t-1} \eqdef \{x_d^{1:M}\}_{d=1}^{n_x} \bigcup \{a_d^{1:N}\}_{d=2}^{n_x}$

\item Set $\tau_t(u_{t-1}) = \prod_{d=1}^{n_x} \frac{1}{M} \sum_{i=1}^M w_d^i$

\item Simulate $x_t \sim\kappa_{t}^M(x_t|u_{t-1})$ by running Algorithm~\ref{alg:bs}
\end{enumerate}
\end{definition}
\begin{proposition}[Proper weighting]
The procedure in Definition~\ref{def:nested} generates $(x_t,\tau_t(u_{t-1}))$ that are properly weighted for $\frac{\gamma_t(x_{1:t})}{\gamma_{t-1}(x_{1:t-1})}$.
\end{proposition}
\begin{proof}
The result follows from Theorem~2 in \citet{naessethLS2015nested}.
\end{proof}
\begin{remark}
Note that we can in fact replace Step~1 of Definition~\ref{def:nested} (\smc and BS) with running the \nsmc algorithm itself, \ie Algorithm~\ref{alg:nsmc}, and letting the $w_d \eqdef 1$ in Step~3. This will also yield properly weighted samples as discussed in \citet{naessethLS2015nested}. We will in the experiments show how this can be used to design efficient algorithms by nesting several layers of \smc samplers.
\end{remark}
Compare with the example in Section~\ref{sec:ffbs} and Figure~\ref{fig:ffbs} where we used forward filtering--backward sampling by considering the components of $v_{t,1:d}$ as our target. Instead of exact forward filtering we can use Algorithm~\ref{alg:smc}, and instead of exact backward sampling we can use Algorithm~\ref{alg:bs}, to generate properly weighted samples.

% ======================================================================
%                   Theoretical justification
% ======================================================================
\subsection{Theoretical Justification}\label{sec:theory}
In this section we will provide a central limit theorem that further motivates \nsmc, and show how the asymptotic variance depends on the internal approximation of the exact fully adapted \smc. Furthermore, we provide a result that shows how this asymptotic variance converges to that of the corresponding asymptotic variance of the exact fully adapted \smc method as $M\to\infty$.

% ====================================================================== 
%                               CLT
% ====================================================================== 
\begin{theorem}[Central Limit Theorem]
\label{thm:clt}
Under certain (standard) regularity conditions on the function $\varphi : \setX_t \mapsto \reals$, specified in the supplementary material, we have the following central limit theorem
\begin{align*}
\sqrt{N} \left(\frac{1}{N}\sum_{i=1}^N \varphi(x_{1:t}^i) - \pi_t(\varphi) \right) \convD \N\left(0,\Sigma_t^M(\varphi)\right),
\end{align*}
where $\{x_{1:t}^i\}_{i=1}^N$ are generated by Algorithm~\ref{alg:nsmc} and the asymptotic variance is given by
\begin{align*}
\Sigma_t^M(\varphi) &= \sum_{s=0}^t \sigma_{s,t}^M(\varphi),
\end{align*}
for $\sigma_{s,t}^M(\varphi)$'s defined by
\begin{align*}
\sigma_{t,t}^M(\varphi) &= \pi_t\left(\left(\varphi - \pi_t(\varphi)\right)^2\right),\\
\sigma_{s,t}^M(\varphi) &= \int \Psi_{s,t}^M(x_{1:s};\varphi) \pi_s(x_{1:s}) \myd x_{1:s}, \quad \text{for} ~ 0<s<t,\\
\sigma_{0,t}^M(\varphi) &= \int \frac{\tau_1(u_0)^2}{Z_1^2} \left( \int \left( \varphi(x_{1:t})-\pi_t(\varphi) \right)  \frac{\pi_t(x_{1:t})}{\pi_1(x_{1})} \kappa_{1}^M(x_{1} | u_{0}) \myd x_{1:t} \right)^2 \eta_{0}^M(u_{0}) \myd u_0.
\end{align*}
with 
\begin{align}
&\Psi_{s,t}^M(x_{1:s};\varphi) \eqdef \nonumber\\
&\E_{\eta_{s}^M(u_{s}|x_{1:s})}\left[\frac{Z_s^2}{Z_{s+1}^2}\tau_{s+1}(u_s)^2 \left( \int \left( \varphi(x_{1:t})-\pi_t(\varphi) \right) \frac{\pi_t(x_{1:t})}{\pi_{s+1}(x_{1:s+1})} \kappa_{s+1}^M(x_{s+1} | u_{s}) \myd x_{s+1:t} \right)^2\right]
\label{eq:condexp}
\end{align}
\end{theorem}
\begin{proof}
See the supplementary material.
\end{proof}

This theorem shows that, even for a fixed and finite value of $M$, the \nsmc method obtains the standard $\sqrt{N}$ convergence rate. %Furthermore, note that the expression of the asymptotic variance is much more informative than the central limit theorem of \citet{naessethLS2015nested}. 
We can see how the asymptotic variance depends on how well we approximate $q_t$ and its normalization constant with $\kappa_t$ and $\tau_t$. Furthermore, this lets us study convergence of the variance in $M$ and also analytic expressions for a high-dimensional state space model.

% ====================================================================== 
%                       Convergence to FAPF
% ======================================================================
To show the convergence to fully adapted \smc as the approximation improves with increasing $M$ we make some further assumptions detailed below.

% UI assumption
\begin{assumption}[Uniform integrability]
The sequence (in $M$) of random variables $\{\Psi_{s,t}^M(x_{1:s};\varphi)\}$ is uniformly integrable.
\label{ass:ui}
\end{assumption}
\begin{remark}
Note that a sufficient condition for Assumption~\ref{ass:ui} to hold is that for some $\delta > 0$ and for all $s, M \geq 1$ the following holds
\begin{align*}
\int \Psi_{s,t}^M(x_{1:s};\varphi)^{1+\delta} \pi_s(x_{1:s}) \myd x_{1:s} < \infty.
\end{align*}
\end{remark}
%Furthermore, $\sigma_0^M(\varphi)<\infty$ for all $M \geq 1$.

% Strong mixing-type assumption
\begin{assumption}[Strong mixing]
\label{ass:strong}
For all $s$, $t$, there exists
\begin{align*}
\lambda_{s+1,t}^{-} \cdot \pi_t(x_{s+2:t}|x_{1:s+1}) \leq \frac{\pi_t(x_{1:t})}{\pi_{s+1}(x_{1:s+1})} \leq \lambda_{s+1,t}^{+} \cdot \pi_t(x_{s+2:t}|x_{1:s+1}),
\end{align*}
where $0 < \lambda_{s+1,t}^-, \lambda_{s+1,t}^+ < \infty$.
\end{assumption}
\begin{remark}
In the supplementary material we detail a weaker assumption for which Proposition~\ref{prop:convfapf} still holds.
\end{remark}

% Convergence proposition
\begin{proposition}
\label{prop:convfapf}
Under the assumptions of Theorem~\ref{thm:clt}, Assumption~\ref{ass:ui} and \ref{ass:strong} the following limit holds:
\begin{align*}
&\lim_{M\to \infty} \Sigma_t^M(\varphi) = \\
&=\pi_t\left(\left(\varphi - \pi_t(\varphi)\right)^2\right) + \sum_{s=1}^{t-1} \int \frac{\pi_t(x_{1:s})^2}{\pi_s(x_{1:s})} \left(\int \varphi(x_{1:t}) \pi_t(x_{s+1:t}|x_{1:s}) \myd x_{s+1:t}-\pi_t(\varphi)\right)^2 \myd x_{1:s}.
\end{align*}
\end{proposition}
\begin{proof}
See the supplementary material.
%By Assumption~\ref{ass:conv} we have that $\sigma_{0,t}(\varphi)\convAS 0$; and together with Assumption~\ref{ass:ui} it follows that the convergence in \eqref{eq:PsiConvAS} holds in distribution \citep{goggin1994}. Thus by the Vitali Convergence Theorem \citep{TODO} it follows that as $M \to \infty$
%\begin{align*}
%\sigma_{s,t}^M(\varphi) \convAS \int \frac{\pi_t(x_{1:s})^2}{\pi_s(x_{1:s})} \left(\int \varphi(x_{1:t}) \pi_t(x_{s+1:t}|x_{1:s}) \myd x_{s+1:t}-\pi_t(\varphi)\right)^2 \myd x_{1:s}.
%\end{align*}
\end{proof}
\begin{remark}
The attained asymptotic variance is exactly the one derived for the fully adapted SMC asymptotic variance by \citet{johansen2008}.
\end{remark}

% ======================================================================
%                   Choosing N vs M
% ======================================================================
\subsection{Choosing $N$ vs $M$}\label{sec:nvsm}
The computational complexity for the two-level \nsmc is proportional to $\mathcal{O}(NM)$, and it is interesting to study the trade-off between the number of particles in the outer procedure ($N$) and the inner ($M$). To this end we consider a fairly simple model and test function that leads to analytical expressions for the asymptotic variance in the CLT above. We propose to study a high-dimensional \ssm, given in Definition~\ref{def:indepmodel}, \ie obtained by making $n_x$ independent copies of an \ssm. For this model we can obtain analytical solutions given by Proposition~\ref{prop:nvsm}.
\begin{definition}
\label{def:indepmodel}
Define the independent state space model as follows
\begin{align*}
\pi_t(x_{1:t}) \propto \prod_{d=1}^{n_x} \left[\mu(x_{1,d}) \prod_{s=1}^t g(y_{s,d}|x_{s,d}) \prod_{s=2}^t f(x_{s,d}|x_{s-1,d}) \right].
\end{align*}
For simplicity we also assume that $y_{s,d} = y_{s,e}, \forall d,e$ and that $\E_{\pi_t}[x_t] = 0$.
\end{definition}
\begin{proposition}[$N$ vs $M$]
\label{prop:nvsm}
For the model in Definition~\ref{def:indepmodel} and $\varphi(x_{1:t}) = \sum_{d=1}^{n_x} x_{t,d}$, we have that the asymptotic variance of fully adapted \smc is given by
\begin{align*}
\Sigma_t^{\text{FA}}(\varphi) = n_x A_t + \sum_{s=1}^{t-1} n_x B_s^{n_x-1} A_s + n_x(n_x-1) B_s^{n_x-2} C_s^{2},
\end{align*}
and using $r(x_{s,d}|x_{s-1,d})$ as proposal in the \nsmc method in Definition~\ref{def:nested} we get that the asymptotic variance of \nsmc is
\begin{align*}
\Sigma_t^M(\varphi) &= n_x A_t + \sum_{s=0}^{t-1} \Bigg[ n_x B_s^{n_x-1} \left( A_s + M^{-1}\left(\tilde A_s - A_s \right) \right) \left(1 - \frac{1}{M} \right)^{n_x-1}\left(1 + \frac{\tilde B_s}{B_s (M-1)} \right)^{n_x-1} \\
& +n_x(n_x-1) B_s^{n_x-2} \left(C_s + M^{-1}\left( \tilde C_s - C_s\right) \right)^2  \left(1 - \frac{1}{M} \right)^{n_x-2}\left(1 + \frac{\tilde B_s}{B_s (M-1)} \right)^{n_x-2} \Bigg],
\end{align*}
for the (finite) positive constants $A_t, A_s, \tilde A_s, B_s, \tilde B_s, C_s$, and $\tilde C_s$ defined in the supplementary material.
\end{proposition}
\begin{proof}
See the supplementary material.
\end{proof}
\begin{remark}
As expected the asymptotic variance of the \nsmc will (like fully adapted \smc) in general scale exponentially bad with the dimension $n_x$ of the state. However, to control the additional approximation introduced by not evaluating $\nu_{t-1}$ and sampling $q_t$ exactly, we only need to scale $M \propto n_x$, even as $n_x \to \infty$. We expect that intuition and rule-of-thumbs from running standard \smc also apply to the internal approximation targeting $\frac{\gamma_t(x_{1:t})}{\gamma_{t-1}(x_{1:t-1})}$.
\end{remark}

%\cn{Perhaps a numerical example, LGSS with $t=2$?}

% ======================================================================
%                           Experiments
% ======================================================================
\section{Numerical Results}\label{sec:experiments}
% ======================================================================
%                           LGSS
% ======================================================================
\subsection{Gaussian Model}\label{sec:lgss}
We start by considering a Gaussian spatio-temporal state space model where the exact solution is available via the Kalman filter \citep{Kalman:1960}, and we can implement exact fully adapted \smc as explained in Section~\ref{sec:ffbs}. The model is given by
\begin{subequations}
\begin{align}
x_t &= 0.5 x_{t-1}+v_t,&v_t \sim \frac{1}{Z_v} \exp{\left(-\frac{\tau}{2}\sum_{d=1}^{n_x}  v_{t,d}^2-\frac{\lambda}{2} \sum_{d=2}^{n_x} (v_{t,d}-v_{t,d-1})^2\right)}\\
y_t|x_t &\sim \N(x_t,\sigma_y^2 I).
\end{align}
\end{subequations}
The results for $N=100, T=10, \tau=\lambda=1$ and $\sigma_y^2 = 0.25^2$, \ie with fairly high signal to noise ratio, is given in Figure~\ref{fig:lgss}. We compare \nsmc with (and without) backward simulation to the bootstrap particle filter (BPF) that uses the transition probability as proposal. We give all methods equivalent computational budget as the number of internal particles $M$ grow, \ie BPF gets $N_{\text{BPF}}=100\cdot M$ particles. Furthermore, for illustrative purposes we include fully adapted \smc (FAPF), the method that \nsmc approximates, for a fixed number of particles $N_{\text{FAPF}}=100$. The experiments are run ten times independently and we show the median squared error (MSE) as well as 25\%/75\% quantiles, for estimates of the log-likelihood, $\E[x_{T,1}]$ and $\E[x_{T,n_x}]$ with $n_x \in \{10, 100\}$. The expectations are with respect to the posterior distribution.

%\begin{figure}[htbp]
    %\vspace{-5mm}
    %\centering
    %\begin{subfigure}[b]{0.48\textwidth}
        %\includegraphics[width=\textwidth]{{T10D10_sigma0.25_loglikelihood}.pdf}
        %\caption{\footnotesize$n_x=10$}
        %\label{fig:D10:ll}
    %\end{subfigure}
    %~
    %\begin{subfigure}[b]{0.48\textwidth}
        %\includegraphics[width=\textwidth]{{T10D100_sigma0.25_loglikelihood}.pdf}
        %\caption{\footnotesize$n_x=100$}
        %\label{fig:D100:ll}
    %\end{subfigure}
    
    %\begin{subfigure}[b]{0.48\textwidth}
        %\includegraphics[width=\textwidth]{{T10D10_sigma0.25_ex_t9d0}.pdf}
        %\caption{\footnotesize$d=1,n_x=10$}
        %\label{fig:D10:ex0}
    %\end{subfigure}
    %~
    %\begin{subfigure}[b]{0.48\textwidth}
        %\includegraphics[width=\textwidth]{{T10D100_sigma0.25_ex_t9d0}.pdf}
        %\caption{\footnotesize$d=1,n_x=100$}
        %\label{fig:D100:ex0}
    %\end{subfigure}
    
    %\begin{subfigure}[b]{0.48\textwidth}
        %\includegraphics[width=\textwidth]{{T10D10_sigma0.25_ex_t9d9}.pdf}
        %\caption{\footnotesize$d=10, n_x=10$}
        %\label{fig:D10:ex9}
    %\end{subfigure}
    %~
    %\begin{subfigure}[b]{0.48\textwidth}
        %\includegraphics[width=\textwidth]{{T10D100_sigma0.25_ex_t9d99}.pdf}
        %\caption{\footnotesize$d=100,n_x=100$}
        %\label{fig:D100:ex99}
    %\end{subfigure}
    %\caption{\footnotesize Median MSE and 25\%/75\% quantiles of Monte Carlo estimates of $log p(y_{1:T}), \E[x_{T,1}], \E[x_{T,n_x}]$ for BPF, FAPF and two variants of \nsmc. $N=100$ for FAPF and \nsmc and BPF has equivalent computational budget $N=100*M$. Left column $n_x=10$, right column $n_x=100$ and $T=10,\sigma_y^2 = 0.25^2$ in all experiments.}\label{fig:lgss}
%\end{figure}

\begin{figure}[htbp]
    \centering
	\begin{tabular}{m{.02\textwidth} m{.48\textwidth} m{.48\textwidth}  m{.000005\textwidth} }
	 & \centering \small $n_x=10$ & \centering \small $n_x=100$ & \\
	 &\includegraphics[width=.48\textwidth]{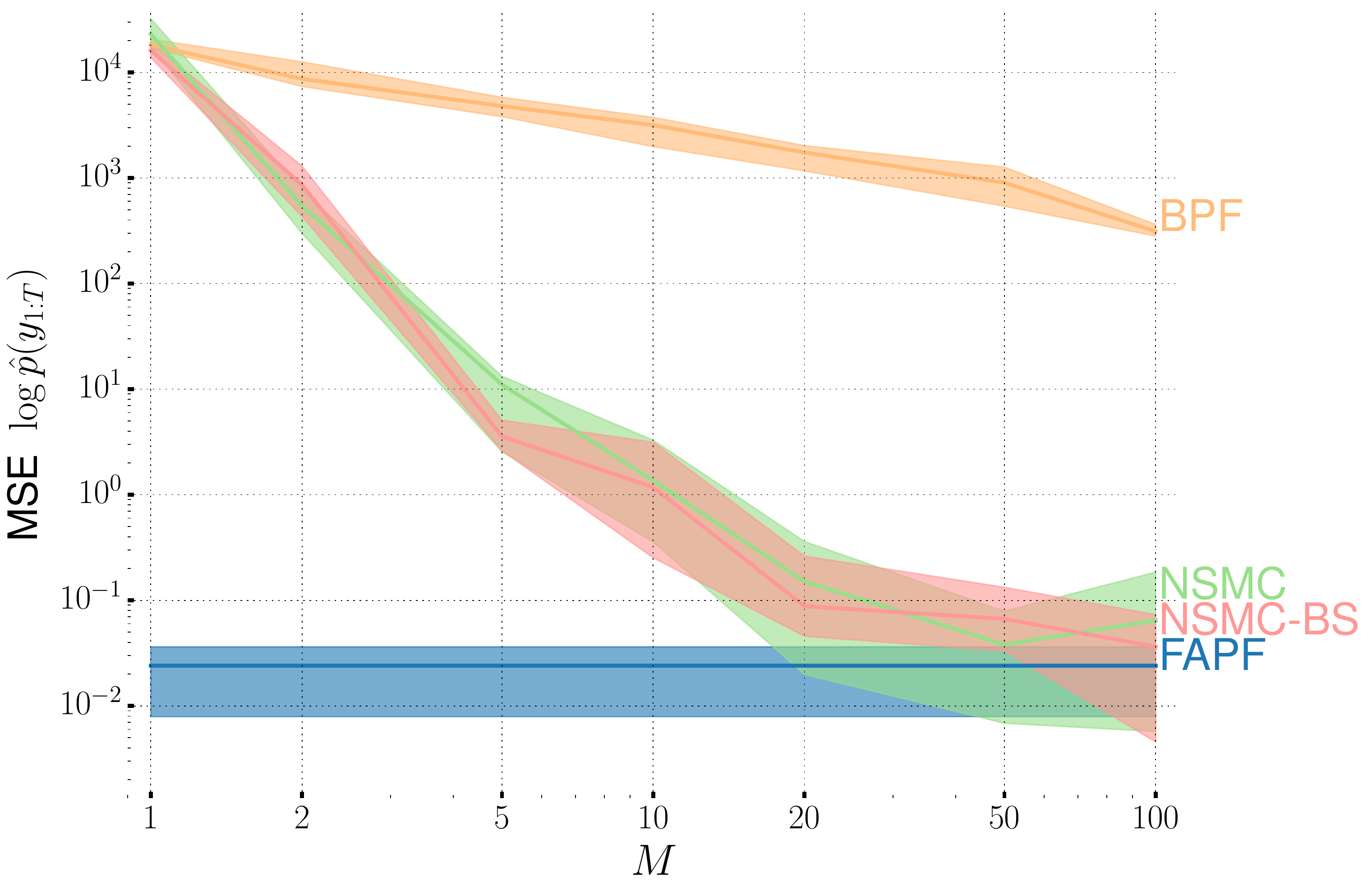} & \includegraphics[width=.48\textwidth]{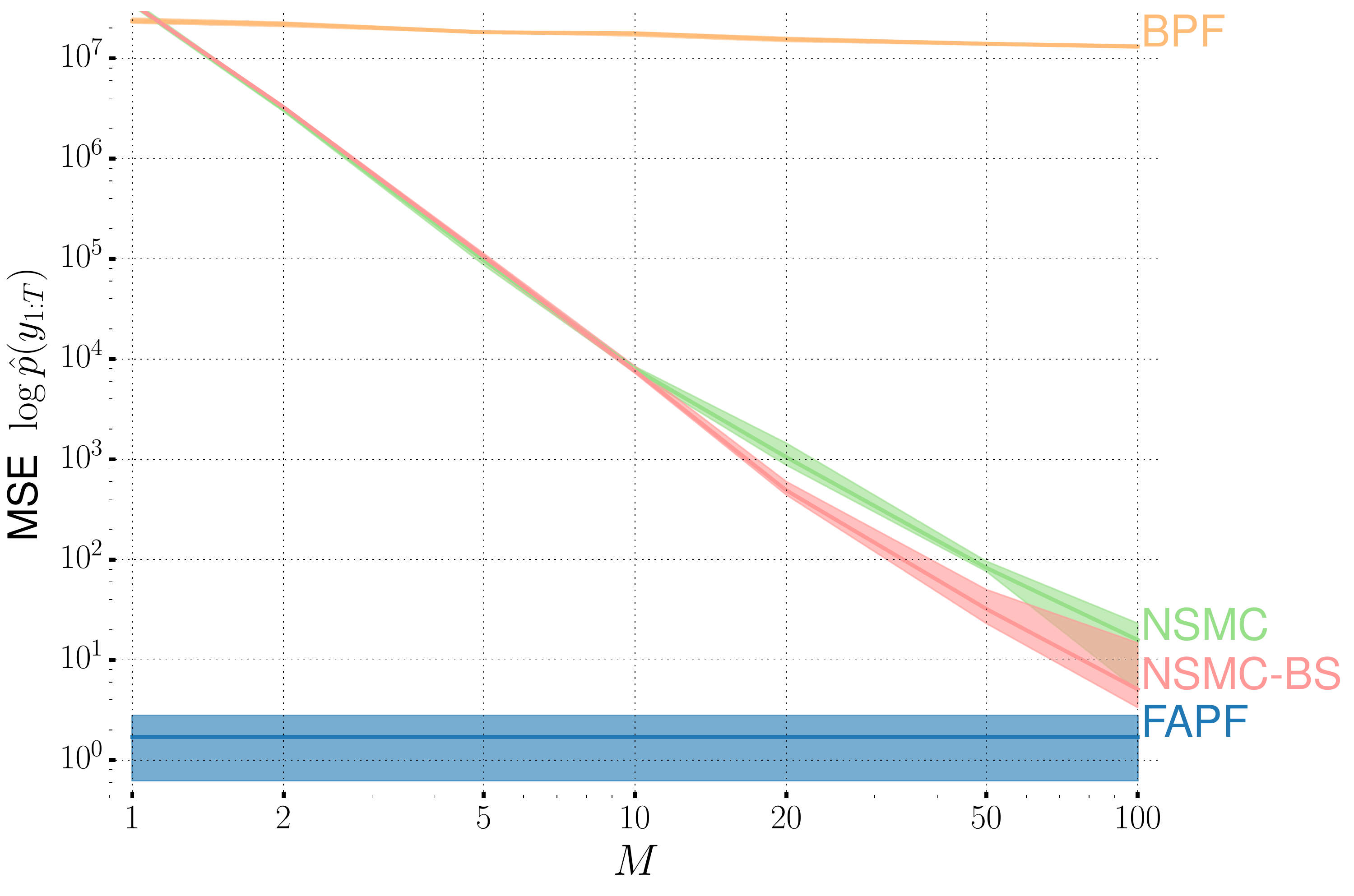} & \\
     \rotatebox{90}{$d=1$} & \includegraphics[width=.48\textwidth]{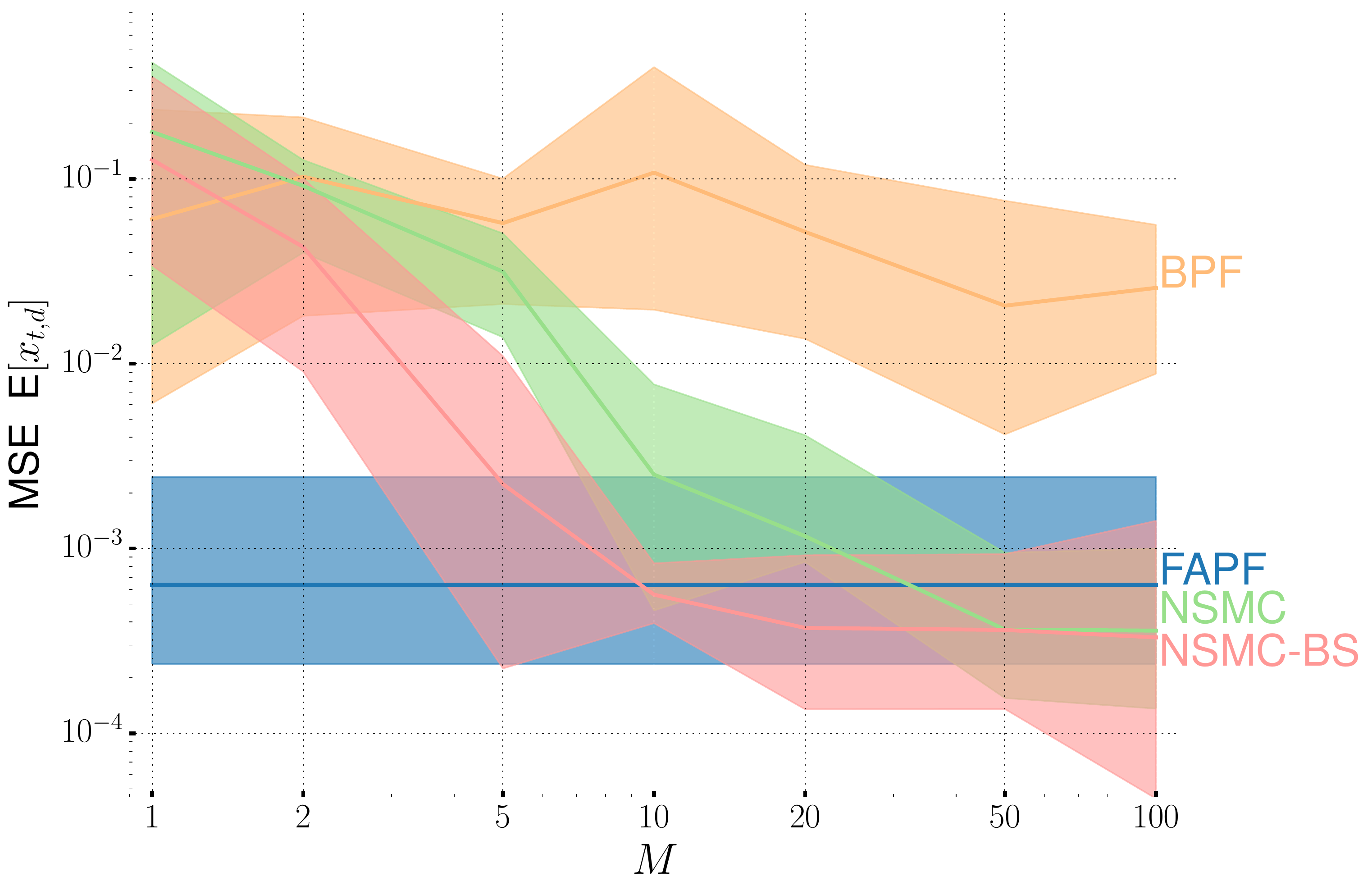} & \includegraphics[width=.48\textwidth]{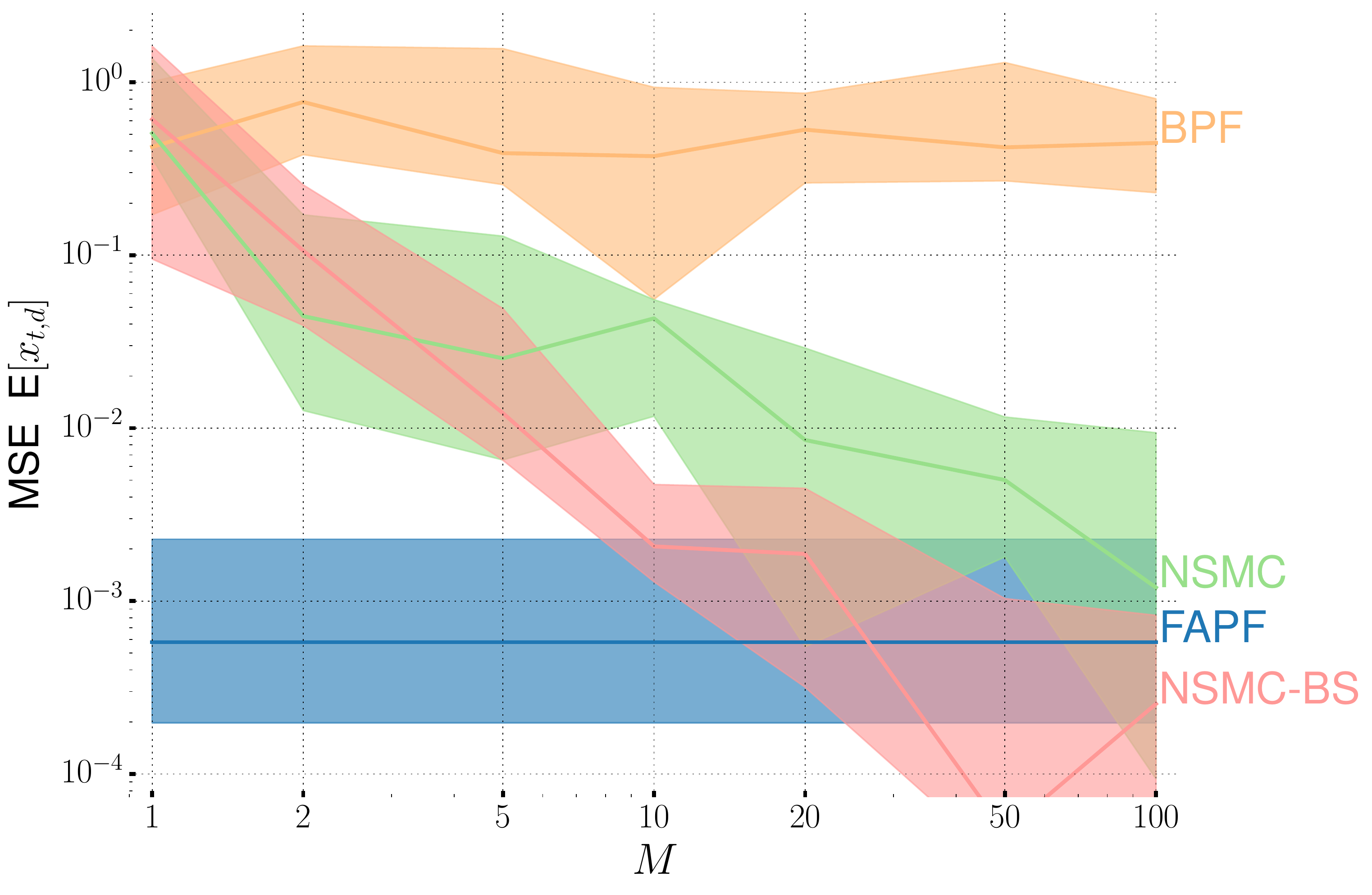} & \\
     \rotatebox{90}{$d=n_x$} &\includegraphics[width=.48\textwidth]{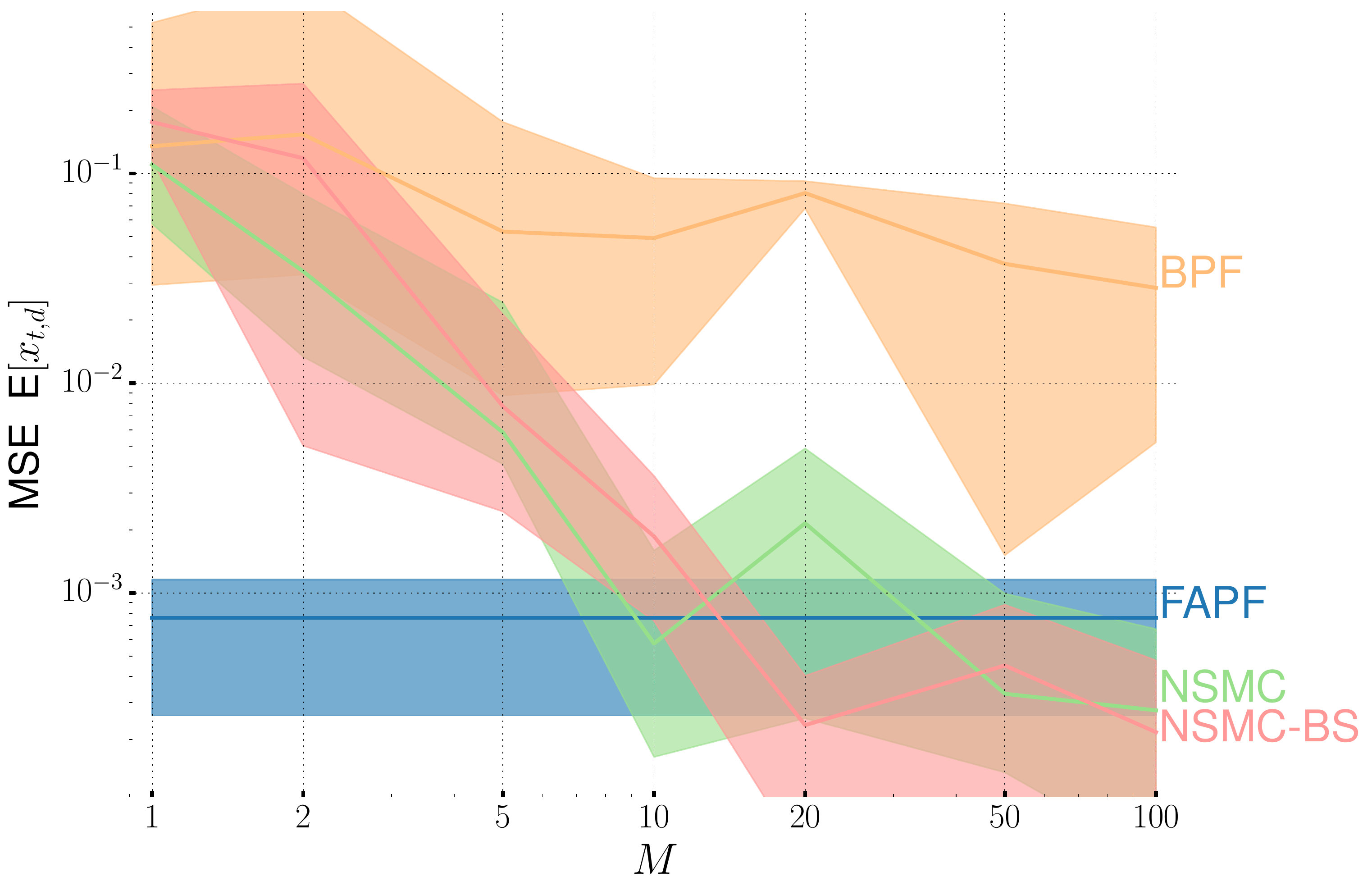} & \includegraphics[width=.48\textwidth]{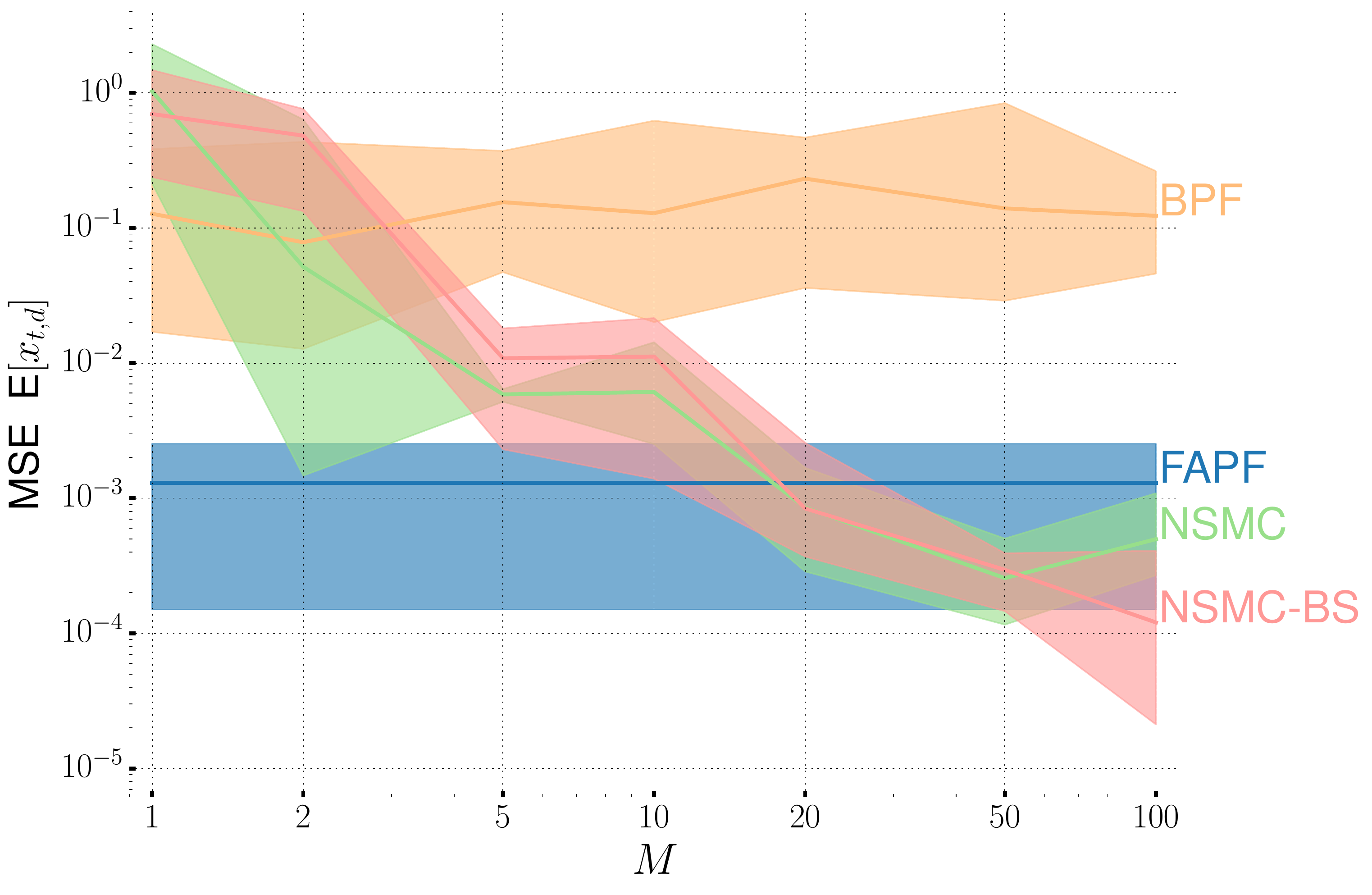} &
	\end{tabular}
	\caption{\footnotesize Median SE and 25\%/75\% quantiles of Monte Carlo estimates of $\log p(y_{1:T}), \E[x_{T,1}], \E[x_{T,n_x}]$ for BPF, FAPF and two variants of \nsmc. $N=100$ for FAPF and \nsmc and BPF has equivalent computational budget $N=100\cdot M$. Left column $n_x=10$, right column $n_x=100$ and $T=10,\sigma_y^2 = 0.25^2$ in all experiments.}\label{fig:lgss}
\end{figure}

We can see that \nsmc is significantly better than BPF and that it converges quickly towards the fully adapted \smc. Backward simulation also clearly helps with estimates of $\E[x_{T,d}]$ for $d=1$, alleviating the correlation between generated samples. It is worthwhile to point out that for small $M$ the \nsmc seems to improve much more quickly than the standard asymptotic rate $M^{-1}$. For the likelihood estimate the rate almost exceeds $M^{-4}$. We provide results for different settings of $\sigma_y^2$ in the supplementary material. In general we see less striking improvement of \nsmc over BPF when the signal to noise ratio is low, \ie $\sigma_y^2$ is high compared to $\tau^{-1}$, which is to be expected \citep{snyder2015performance}.

\subsection{Soil Carbon Cycles}\label{sec:nlss}
We move on to study the performance of \nsmc and compare it to \stpf on a spatio-temporal model inspired by the soil carbon cycle model of \citep{murray2016,clifford2014}. The simplified model that we use to profile the two state-of-the-art methods is defined by
\begin{subequations}
\begin{align}
x_t &= 0.5(x_{t-1}+e^{\xi_t}) e^{v_t}, \qquad v_t \sim \frac{1}{Z_v}\exp\Big(-\frac{\tau}{2}\sum_{i \in \Ve} v_{t,i}^2 - \frac{\lambda}{2}\sum_{(i,j)\in \Ed} (v_{t,i}-v_{t,j})^2\Big),\\
y_t | x_t & \sim \operatorname{Truncated Normal}\left( x_t, \sigma^2 I , 0, \infty\right),
\end{align}
\label{eq:soilcarbon}
\end{subequations}
where $\xi_t$ is a known input signal and $(\Ve,\Ed)$ is a square lattice, $\sqrt{n_x} \times \sqrt{n_x}$, with nearest neigbour interaction, \ie $(i,j) \in \Ed$ if $i$ and $j$ are neighbors on the lattice. The latent variables $x_t$ are positive and it is not possible to implement the exact fully adapted \smc method. We set $\sigma=0.2$, $\tau = 2$, and $\lambda = 1.0$ and run \nsmc and \stpf with matched computational complexity. Figure~\ref{fig:nlss} displays the median, over the $n_x$ dimensions, mean squared error for each time-point $t$ estimated by running the algorithms $20$ times independently. Ground truth is estimated using $20$ independent runs of the method of \citet{naessethLS2014} with $\thsnd{64}$ samples.
\begin{figure}[htbp]
    \centering
	\includegraphics[width=.48\textwidth]{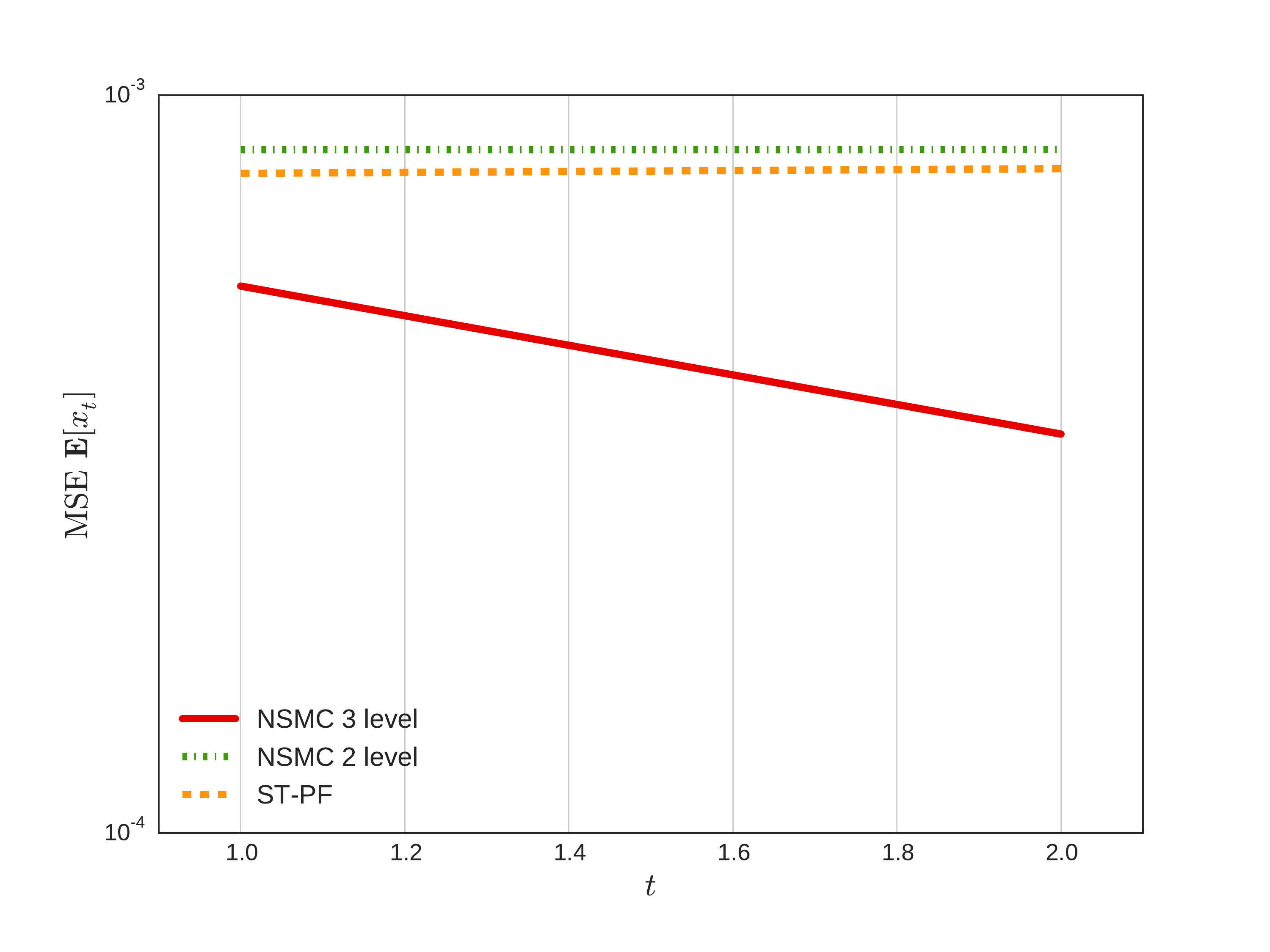}
    \includegraphics[width=.48\textwidth]{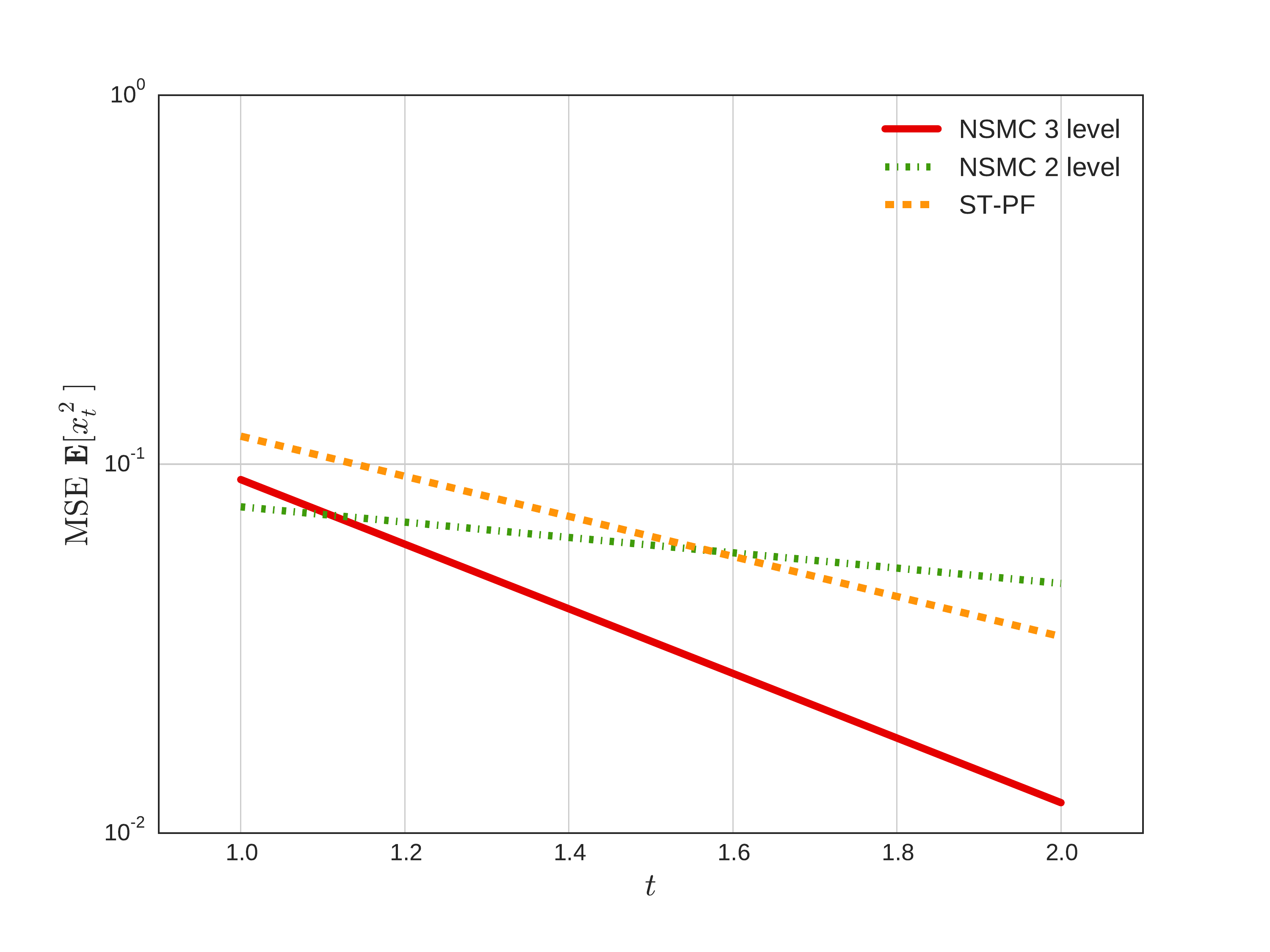}
	\caption{Results for $T=2$, $n_x = 64$ ($8\times 8$).}\label{fig:nlss}
\end{figure}
We can see that the different \nsmc versions either perform as well, or better than \stpf. This is without taking into account that \nsmc simplifies distribution of the computation and is more memory efficient, only $N$ rather than $NM$ samples need to be retained at each step.

\subsection{Mixture Model}\footnote{The results in this section have been previously published by the authors in \citet{naessethLS2015nested}.}
Finally, we consider an example with a non-Gaussian \stssm, borrowed from \citet{beskosCJKZ2014a} where the full details of the model are given.
The transition probability $f(x_t \mid x_{t-1})$ is a spatially localised Gaussian mixture and the measurement probability $g(y_t \mid x_t)$ is Student's t-distributed. The model dimension is $n_x=1\thinspace024$. \citet{beskosCJKZ2014a} report improvements for \stpf over both the BPF 
and the block PF by \citet{rebeschiniH2015can}. Following \citet{beskosCJKZ2014a} we use $N = M = 100$ for both \stpf and \nsmc 
\begin{figure}[h]
\begin{center}
    \includegraphics[width=0.5\columnwidth]{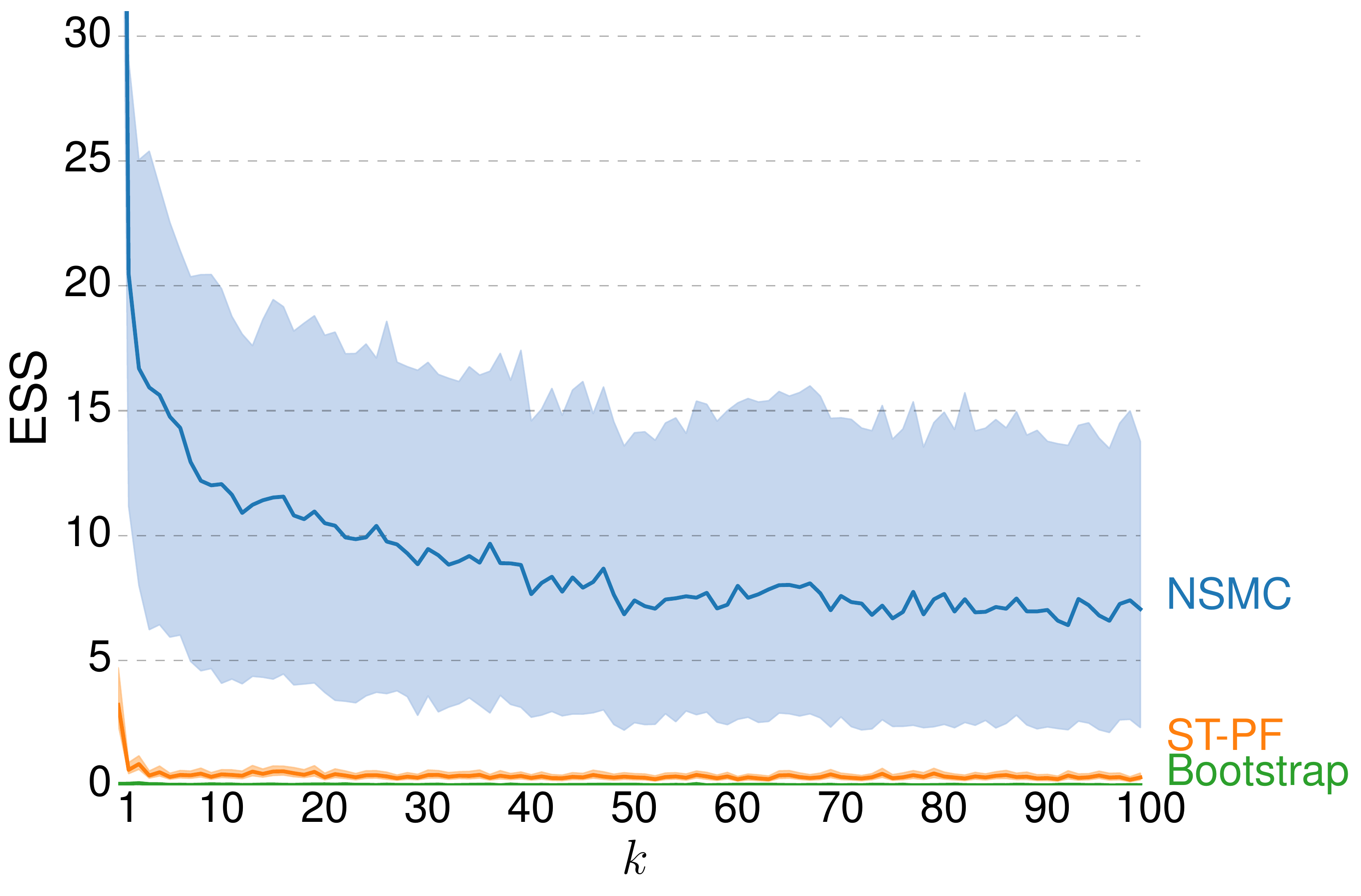}
  \end{center}
\caption{Median \ess with $15-85\%$ percentiles (shaded region) for the non-Gaussian \ssm.}
\label{fig:nless}
\end{figure}
and the BPF is given $N=\thsnd{10}$. In Figure~\ref{fig:nless} we report the effective sample size (\ess, higher is better), estimated according to \citet{CarpenterCF:1999}. The \ess for the BPF is close to $0$, for \stpf around 1--2, and for \nsmc slightly higher at 7--8. However, we note that all methods perform quite poorly on this model, and to obtain satisfactory results it would be necessary to use more particles.

% ======================================================================
%                           Appendix/Supp
% ======================================================================
\appendix
\section{Supplementary Material}

% ======================================================================
%                             Method
% ======================================================================
\subsection{General Nested Sequential Monte Carlo}\label{sec:method}
Assume that we are interested in approximating an arbitrary auxiliary \smc sampler with proposal $q_t(x_t|x_{1:t-1}) = \frac{r_t(x_t|x_{1:t-1})}{\int r_t(x_t|x_{1:t-1}) \myd x_t}$ and \emph{adjustment multipliers} $\nu_{t-1}(x_{1:t-1})$. The fully adapted \smc that we focus on in this paper is then attained as a special case when $q_t(x_t|x_{1:t-1}) \propto \frac{\gamma_t(x_{1:t})}{\gamma_{t-1}(x_{1:t-1})}$ and $\nu_{t-1}(x_{1:t-1}) = \int \frac{\gamma_t(x_{1:t})}{\gamma_{t-1}(x_{1:t-1})} \myd x_t$. 

We can just as easily use a nested Monte Carlo method that produces properly weighted samples with respect to an arbitrary proposal $q_t$ and multipliers $\nu_{t-1}$, see Algorithm~\ref{alg:nsmc}.

\begin{algorithm}
\caption{Nested Sequential Monte Carlo\hfill (all for $i=1,\ldots,N$)}\label{alg:nsmc}
\begin{algorithmic}[1]
\REQUIRE $\eta_{t-1}^M, \kappa_t^M, \tau_t$ that generate samples properly weighted for $q_t(x_t|x_{1:t-1})$% \propto r_t(x_t|x_{1:t-1})$
%\STATE $x_0^i \sim \mu(x_0)$
\FOR{$t=1$ to $T$}
\STATE Simulate $u_{t-1}^i \sim \eta_{t-1}^M(u_{t-1}|x_{1:t-1}^i)$
\STATE Draw $a_t^i$ with probability $\Prb(a_t^i = j) = \frac{\hat\nu_{t-1}(x_{1:t-1}^j,u_{t-1}^j) w_{t-1}^j}{\sum_\ell \hat\nu_{t-1}(x_{1:t-1}^\ell,u_{t-1}^\ell) w_{t-1}^\ell}$
\STATE Simulate $x_t^i \sim \kappa_t^M(x_t | u_{t-1}^{a_t^i})$
\STATE Set $x_{1:t}^i = (x_{1:t-1}^{a_t^i},x_t^i)$
\STATE Set $w_t^i = \frac{\gamma_t(x_{1:t}^i)}{\gamma_{t-1}(x_{1:t-1}^{a_t^i})} \frac{\tau_t(u_{t-1}^{a_t^i})}{\hat\nu_{t-1}(x_{1:t-1}^{a_t^i},u_{t-1}^{a_t^i}) r_t(x_t^i|x_{1:t-1}^{a_t^i})}$
\ENDFOR
\end{algorithmic}
\end{algorithm}

\begin{remark}
Note that if the adjustment multipliers $\nu_{t-1}$ do not depend on $u_{t-1}$, simulating from $\eta_{t-1}$ can be done after resampling (simulating $a_t$). This ensures that the new samples are conditionally independent, thus decreasing correlation between samples.
\end{remark}

\subsubsection*{Generating Properly Weighted Samples using \is}
There are many ways of generating properly weighted samples with respect to a distribution, one example is using sequential Monte Carlo with or without backward simulation as explained in the main manuscript. However, perhaps one of the most straightforward and simple approaches is to use standard importance sampling. This means we would define $\eta_{t-1}^M, \kappa_t^M, \tau_t$ as follows:
\begin{description}
\item{$\eta_{t-1}^M(u_{t-1}|x_{1:t-1})$:} Set $u_{t-1} = \{\tilde x_t^i\}_{i=1}^M$, where $\tilde x_t^i \sim p_t(x_t|x_{1:t-1})$ for some proposal $p_t$,

\item{$\kappa_t^M(x_t|u_{t-1})$:} Set $x_t = \tilde x_t^B$, where $B$ is simulated with probability $\Prb(B = j) = \frac{\tilde w_t^j}{\sum_\ell \tilde w_t^\ell}$ with $w_t^j = \frac{r_t(\tilde x_t^j|x_{1:t-1})}{p_t(\tilde x_t^j|x_{1:t-1})}$,

\item{$\tau_t(u_{t-1})$:} Set $\tau_t(u_{t-1}) = \frac{1}{M} \sum_{i=1}^M \tilde w_t^i$.
\end{description}
It is straightforward to show that the above procedure generates properly weighted samples for $q_t$ as long as $p_t>0$ whenever $q_t$ is. Now, if we want to use the above to approximate fully adapted \smc we simply let $r_t = \gamma_t/\gamma_{t-1}$ and $\hat\nu_{t-1} = \tau_t$.

%\newpage
% ======================================================================
%                             Theory
% ======================================================================
\subsection{Theoretical Results}\label{sec:theory}
%\cn{Change so numbers from main are note hard-coded}
% ======================================================================
%                           THEOREM 1
% ======================================================================
\subsubsection*{Proof of Theorem~\ref{thm:clt}}
We reproduce the central limit theorem of \citet{naessethLS2015nested} here for clarity, see the Appendix of the extended version \citet{naessethLS2015nestedARXIV} for details. %Assume that we are approximating an arbitrary auxiliary \smc sampler with proposal $q_t(x_t|x_{1:t-1}) = \frac{r_t(x_t|x_{1:t-1})}{\int r_t(x_t|x_{1:t-1}) \myd x_t}$ and \emph{adjustment multipliers} $\nu_{t-1}(x_{1:t-1})$. The fully adapted \smc that we focus on in this paper is then attained as a special case when $q_t(x_t|x_{1:t-1}) \propto \frac{\gamma_t(x_{1:t})}{\gamma_{t-1}(x_{1:t-1})}$ and $\nu_{t-1}(x_{1:t-1}) = \int \frac{\gamma_t(x_{1:t})}{\gamma_{t-1}(x_{1:t-1})} \myd x_t$. The complete is summarized in A

% NOTATION
\subsubsection*{Notation and Definitions}
To explicitly state the general theorem we need some notation defined below:
\begin{align*}
\Gamma_t(x_{1:t},u_{0:t}) &= \frac{\tau_t(u_{t-1}) \eta_t^M(u_t|x_{1:t}) \kappa_t^M(x_t|u_{t-1})}{r_t(x_t|x_{1:t-1})} \frac{\gamma_t(x_{1:t})}{\gamma_{t-1}(x_{1:t-1})} \Gamma_{t-1}(x_{1:t-1},u_{0:t-1}), \\
\Pi_t(x_{1:t},u_{0:t}) &= \frac{\Gamma_t(x_{1:t},u_{0:t})}{Z_t},\\
w_t(x_{1:t},u_{0:t}) &\propto \frac{\gamma_t(x_{1:t})}{\gamma_{t-1}(x_{1:t-1})} \frac{\tau_t(u_{t-1})}{\nu_{t-1}(x_{1:t-1},u_{t-1}) r_t(x_t|x_{1:t-1})},\\
\Gamma_t'(x_{1:t},u_{0:t}) &= \nu_t(x_{1:t},u_t)\Gamma_t(x_{1:t},u_{0:t}),\\
\Pi_t'(x_{1:t},u_{0:t}) &\propto \Gamma_t'(x_{1:t},u_{0:t}),\\
Q_t^M(x_t,u_t|x_{1:t-1},u_{t-1}) &= \eta_t^M(u_t|x_{1:t}) \kappa_t^M(x_t|u_{t-1}), \\
w_t'(x_{1:t},u_{0:t}) &= \frac{\Pi_t'(x_{1:t},u_{0:t})}{Q_t^M(x_t,u_t|x_{1:t-1},u_{t-1})\Pi_{t-1}'(x_{1:t-1},u_{0:t-1})} \propto \nu_t(x_{1:t},u_t) w_t(x_{1:t},u_{0:t}),\\
\omega_t(x_{1:t},u_{0:t}) &= \frac{\Pi_t(x_{1:t},u_{0:t})}{Q_t^M(x_t,u_t|x_{1:t-1},u_{t-1})\Pi_{t-1}'(x_{1:t-1},u_{0:t-1})} \propto w_t(x_{1:t},u_{0:t}).
\end{align*}

Domain of $\Pi_t(x_{1:t},u_{0:t})$ is denoted by $\Theta_t = \setX_t \times \mathsf{U}_t$. For a function $h:\setX_t \mapsto \reals$, we define the extension of $h$ to $\Theta_t$ by $h^e(x_{1:t},u_{0:t}) \eqdef h(x_{1:t})$. Let $\Phi_t$ be defined recursively to be the set of measurable functions $h : \Theta_t \mapsto \reals$ such that there exists a $\delta > 0$ with $\E_{Q_t^M \Pi_{t-1}'}[\|w_t' h\|^{2+\delta}] < \infty$, and such that $(x_{1:t-1},u_{0:t-1}) \mapsto \E_{Q_t^M}[w_t' h]$ is in $\Phi_{t-1}$. We are now ready to state the more general central limit theorem of \citet{naessethLS2015nested}.

\begin{theorem}[Central Limit Theorem]
\label{thm:gen:clt}
Assume that $\varphi : \setX_t \mapsto \reals$ is a function such that ${\E_{Q_t^M \Pi_{t-1}'}[\|w_t' \varphi^e\|^{2+\delta}] < \infty}$ for some $\delta > 0$, and that $(x_{1:t-1},u_{0:t-1}) \mapsto \E_{Q_t^M}[\omega_t \varphi^e]$ is in $\Phi_{t-1}$. Then we have the following central limit theorem
\begin{align*}
\sqrt{N} \left(\sum_{i=1}^N \frac{w_t^i}{\sum_{\ell=1}^N w_t^\ell}\varphi(x_{1:t}^i) - \pi_t(\varphi) \right) \convD \N\left(0,\Sigma_t^M(\varphi)\right),
\end{align*}
where $\{(w_t^i,x_{1:t}^i)\}_{i=1}^N$ are generated by Algorithm~5 in \citet{naessethLS2015nestedARXIV} and the asymptotic variance is given by
\begin{align*}
\Sigma_t^M(\varphi) &= \tilde V_t^M(\omega_t(\varphi^e- \E_{\Pi_t}[\varphi^e])),
\end{align*}
where $\tilde V_t^M$ is defined by the following set of recursions for measurable functions $h : \Theta_t \mapsto \reals$ 
\begin{align*}
\tilde V_t^M(h) &= \hat V_{t-1}^M \left( \E_{Q_t^M}[h] \right) + \E_{\Pi_{t-1}'}\left[ \var_{Q_t^M}(h) \right], &t>0, \\
V_t^M(h) &= \tilde V_t^M \left( w_t' (h-\E_{\Pi_t'}[h])\right), & t \geq 0, \\
\hat V_t^M(h) &= V_t^M(h) + \var_{\Pi_t'}(h), &t \geq 0.
\end{align*}
initialized by $\tilde V_0^M(h) = \var_{\eta_0^M}(h)$ for $h : \Theta_0 \mapsto \reals$.
\end{theorem}

% NEW STUFF
\subsubsection*{Approximating the Fully Adapted SMC}
When we are approximating the fully adapted \smc, \ie when we have $q_t(x_{t}|x_{1:t-1})\propto \frac{\gamma_t(x_{1:t})}{\gamma_{t-1}(x_{1:t-1})}$ and $\nu_t(x_{1:t},u_t) = \tau_{t+1}(u_t)$, we can make significant simplifications of the expressions in the general central limit theorem above. Specifically we get that
\begin{align*}
\Pi_t'(x_{1:t},u_{0:t}) &= \frac{\tau_{t+1}(u_t)}{Z_{t+1}} \Gamma_t(x_{1:t},u_{0:t}),\\
w_t'(x_{1:t},u_{0:t}) &= \frac{Z_t}{Z_{t+1}} \tau_{t+1}(u_t),\\
\omega_t(x_{1:t},u_{0:t}) & = 1.
\end{align*}

% Lemma
\begin{lemma}
\label{lem:sum}
The asymptotic variance $\Sigma_t^M(\varphi)$ in Theorem~\ref{thm:gen:clt} when approximating the fully adapted \smc is given by
\begin{align}
\Sigma_t^M(\varphi) = \var_{\eta_0^M}(h_0) + \sum_{s=1}^t \var_{\Pi_{s-1}',Q_s^M}(h_s),
\end{align}
for $h_s$ defined by
\begin{align*}
h_t &= \varphi^e - \E_{\Pi_t}[\varphi^e],\\
h_s &= \frac{Z_s}{Z_{s+1}} \tau_{s+1}(u_s)\left(\E_{Q_{s+1}^M}[h_{s+1}] - \E_{\Pi_s'}\left[ \E_{Q_{s+1}^M}[h_{s+1}]\right] \right), & 1 \leq  s \leq t-1,\\
h_0 &= \frac{1}{Z_1} \tau_1(u_0) \left(\E_{Q_1^M}[h_1] - \E_{\Pi_0'}\left[ \E_{Q_1^M}[h_1]\right] \right),
\end{align*}
where $\Pi_0'(u_0) = \frac{\tau_1(u_0)}{Z_1} \eta_0^M(u_0)$.
\end{lemma}
\begin{proof}
For a function $h_t : \Theta_t \mapsto \reals$ we have by Theorem~\ref{thm:gen:clt} that
\begin{align*}
\tilde V_t^M(h_t) &= \hat V_{t-1}^M \left( \E_{Q_t^M}[h_t] \right) + \E_{\Pi_{t-1}'}\left[ \var_{Q_t^M}(h_t) \right] \\
&= V_{t-1}^M\left( \E_{Q_t^M}[h_t] \right)+ \var_{\Pi_{t-1}'}\left(\E_{Q_t^M}[h_t]\right) +\E_{\Pi_{t-1}'}\left[ \var_{Q_t^M}(h_t) \right] \\
&= \tilde V_{t-1}^M \left( w_{t-1}' (\E_{Q_t^M}[h_t] -\E_{\Pi_{t-1}'}[\E_{Q_t^M}[h_t] ])\right)+ \var_{\Pi_{t-1}'}\left(\E_{Q_t^M}[h_t]\right) +\E_{\Pi_{t-1}'}\left[ \var_{Q_t^M}(h_t) \right] \\
&= \ldots = \tilde V_{t-1}^M \left(\frac{Z_{t-1}}{Z_t} \tau_t(u_{t-1}) \left(\E_{Q_t^M}[h_t] -\E_{\Pi_{t-1}'}[\E_{Q_t^M}[h_t] ] \right)\right) + \var_{\Pi_{t-1}',Q_t^M}\left(h_t\right)
\end{align*}
Recursion with $h_{t-1} \eqdef \frac{Z_{t-1}}{Z_t} \tau_t(u_{t-1}) \left(\E_{Q_t^M}[h_t] -\E_{\Pi_{t-1}'}[\E_{Q_t^M}[h_t] ] \right)$ gives the result.
\end{proof}

% Lemma
\begin{lemma}
\label{lem:hs}
\begin{align}
h_t &= \varphi - \pi_t(\varphi),\\
h_s &= \frac{Z_s}{Z_{s+1}} \tau_{s+1}(u_s) \int \left(\varphi(x_{1:t})-\pi_t(\varphi) \right) \frac{\pi_t(x_{1:t})}{\pi_{s+1}(x_{1:s+1})} \kappa_{s+1}^M(x_{s+1}|u_s) \myd x_{s+1:t}, & 1 \leq  s \leq t-1,\label{eq:hs}\\
h_0 &= \frac{1}{Z_1} \tau_1(u_0) \int \left(\varphi(x_{1:t})-\pi_t(\varphi) \right) \frac{\pi_t(x_{1:t})}{\pi_{1}(x_{1})} \kappa_{1}^M(x_{1}|u_0) \myd x_{1:t},
\end{align}
\end{lemma}
\begin{proof}
The first, $h_t$, follows straightforwardly by the definition of $\varphi^e$ and $\Pi_t$. The remaining will be proved by induction. Assume that for $s \leq t-1$ \eqref{eq:hs} holds. We will now show that this in fact holds for both $h_{t-1}$ and $h_{s-1}$; thus the result follows. Start by considering $h_{t-1}$ using the definition in Lemma~\ref{lem:sum}
\begin{align*}
h_{t-1} &= \frac{Z_{t-1}}{Z_{t}} \tau_{t}(u_{t-1})\left(\E_{Q_{t}^M}[h_{t}] - \E_{\Pi_{t-1}'}\left[ \E_{Q_{t}^M}[h_{t}]\right] \right) = \frac{Z_{t-1}}{Z_{t}} \tau_{t}(u_{t-1})\left(\E_{Q_{t}^M}\left[\varphi - \pi_t(\varphi)\right] - 0\right)\\
&=\frac{Z_{t-1}}{Z_{t}} \tau_{t}(u_{t-1})\left(\int \varphi(x_{1:t}) \kappa_t^M(x_t|u_{t-1}) \myd x_t - \pi_t(\varphi)\right).
\end{align*}
Now, for $h_{s-1}$ let us start by studying $\E_{Q_{s}^M}[h_{s}]$ and $ \E_{\Pi_{s-1}'}\left[ \E_{Q_{s}^M}[h_{s}]\right]$
\begin{align*}
\E_{Q_{s}^M}[h_{s}] &= \E_{Q_{s}^M}\left[\frac{Z_s}{Z_{s+1}} \tau_{s+1}(u_s) \int \left(\varphi(x_{1:t})-\pi_t(\varphi) \right) \frac{\pi_t(x_{1:t})}{\pi_{s+1}(x_{1:s+1})} \kappa_{s+1}^M(x_{s+1}|u_s) \myd x_{s+1:t}\right] \\
&=\ldots = \int \left(\varphi(x_{1:t}-\pi_t(\varphi)\right)\frac{\pi_t(x_{1:t})}{\pi_{s}(x_{1:s})}\kappa_s^M(x_s|u_{s-1}) \myd x_{s:t},\\
 \E_{\Pi_{s-1}'}\left[ \E_{Q_{s}^M}[h_{s}]\right] &= \ldots = 0.
\end{align*}
This gives us that
\begin{align*}
h_{s-1} &= \frac{Z_{s-1}}{Z_{s}} \tau_{s}(u_{s-1})\left(\E_{Q_{s}^M}[h_{s}] - \E_{\Pi_{s-1}'}\left[ \E_{Q_{s}^M}[h_{s}]\right] \right) \\
&= \frac{Z_{s-1}}{Z_{s}} \tau_{s}(u_{s-1}) \int \left(\varphi(x_{1:t})-\pi_t(\varphi) \right) \frac{\pi_t(x_{1:t})}{\pi_{s}(x_{1:s})} \kappa_{s}^M(x_{s}|u_{s-1}) \myd x_{s:t}.
\end{align*}
The results follows by noting that the procedure is the same for $h_0$ taking into account edge effects, \ie $Z_0 = 1$.
\end{proof}

% Lemma
\begin{lemma}
\label{lem:varhs}
\begin{align*}
\var_{\Pi_{t-1}',Q_t^M}(h_t) &= \pi_t\left((\varphi-\pi_t(\varphi))^2\right),\\
\var_{\Pi_{s-1}',Q_s^M}(h_s) &= \int \Bigg[ \frac{Z_s^2 \tau_{s+1}(u_s)^2}{Z_{s+1}^2} \left( \int \left(\varphi(x_{1:t}) -\pi_t(\varphi) \right) \frac{\pi_t(x_{1:t})}{\pi_{s+1}(x_{1:s+1})} \kappa_{s+1}^M(x_{s+1}|u_s) \myd x_{s+1:t} \right)^2 \\
&\quad \eta_s^M(u_s|x_{1:s-1})\pi_s(x_{1:s}) \Bigg] \myd u_s \myd x_{1:s}, \quad\quad 1 \leq s \leq t-1,\\
\var_{\eta_0^M}(h_0) &= \int \frac{\tau_1(u_0)^2}{Z_1^2} \left( \int \left(\varphi(x_{1:t}) -\pi_t(\varphi) \right) \frac{\pi_t(x_{1:t})}{\pi_1(x_1)} \kappa_1^M(x_1|u_0) \myd x_{1:t} \right)^2 \eta_0^M(u_0)\myd u_0
\end{align*}
\end{lemma}
\begin{proof}
We get the first equality
\begin{align*}
\var_{\Pi_{t-1}',Q_t^M}(h_t) &= \E_{\Pi_{t-1}',Q_t^M}\left[ \left( \varphi - \pi_t(\varphi) \right)^2 \right] - \left(\E_{\Pi_{t-1}',Q_t^M}\left[ \varphi - \pi_t(\varphi) \right] \right)^2 = \pi_t\left( (\varphi-\pi_t(\varphi)^2\right),
\end{align*}
due to Lemma~\ref{lem:hs} and because $\Pi_{t-1}'(x_{1:t-1},u_{0:t-1}) Q_t^M(x_t,u_t|x_{1:t-1},u_{0:t-1}) = \Pi_t(x_{1:t},u_{0:t})$.
\begin{align*}
&\var_{\Pi_{s-1}',Q_s^M}(h_s) = \E_{\Pi_{s-1}',Q_s^M}\left[h_s^2\right] - \left( \E_{\Pi_{s-1}',Q_s^M}\left[h_s\right] \right)^2 = \E_{\Pi_{s-1}',Q_s^M}\left[h_s^2\right] \\
&=\int h_s(x_{1:s},u_s)^2 \Pi_s(x_{1:s},u_{0:s}) \myd u_{0:s} \myd x_{1:s} = \int h_s(x_{1:s},u_s)^2 \eta_s^M(u_s|x_{1:s}) \pi_s(x_{1:s}) \myd u_{s} \myd x_{1:s} \\
&= \int \Bigg[ \frac{Z_s^2 \tau_{s+1}(u_s)^2}{Z_{s+1}^2} \left( \int \left(\varphi(x_{1:t}) -\pi_t(\varphi) \right) \frac{\pi_t(x_{1:t})}{\pi_{s+1}(x_{1:s+1})} \kappa_{s+1}^M(x_{s+1}|u_s) \myd x_{s+1:t} \right)^2 \\
&\quad \eta_s^M(u_s|x_{1:s-1})\pi_s(x_{1:s}) \Bigg] \myd u_s \myd x_{1:s}, \quad\quad 1 \leq s \leq t-1,
\end{align*}
where the second equality follows by noting that $\E_{\Pi_{s-1}',Q_s^M}\left[h_s\right] = 0$. Analogously to Lemma~\ref{lem:hs} the expression for $s=0$ follows by taking into account the edge effects.
\end{proof}

Finally, with Lemmas~\ref{lem:sum}, \ref{lem:hs}, and \ref{lem:varhs} together the result, \ie Theorem~1 of the main manuscript, follows.

% ======================================================================
%                           PROPOSITION CONVERGENCE
% ======================================================================
\subsubsection*{Proof of Proposition~\ref{prop:convfapf} in the Main Manuscript}
% Approx. assumption
\begin{assumption}[Approximation property]
\label{ass:conv}
The approximation of $q_s(x_s|x_{1:s-1}) \propto \frac{\pi_s(x_{1:s})}{\pi_{s-1}(x_{1:s-1})}$ and $\nu_{s-1}(x_{1:s-1}) = \int \frac{\pi_s(x_{1:s})}{\pi_{s-1}(x_{1:s-1})} \myd x_s $ based on $\eta_s^M$, $\kappa_{s+1}^M$ and $\tau_s$ satisfies
\begin{align}
\Psi_{s,t}^M(x_{1:s};\varphi) \convD \frac{\pi_t(x_{1:s})^2}{\pi_{s}(x_{1:s})^2} \left( \int \varphi(x_{1:t}) \pi_t(x_{s+1:t}|x_{1:s})\myd x_{s+1:t} - \pi_t(\varphi)\right)^2, ~\text{as}~M \to \infty.
\label{eq:PsiConvD}
\end{align}
Furthermore, assume that $\sigma_{0,t}^M(\varphi) \convD 0$ as $M \to \infty$.
\end{assumption}
\begin{lemma}
The strong mixing assumption of the main manuscript,
\begin{align*}
\lambda_{s+1,t}^{-} \cdot \pi_t(x_{s+2:t}|x_{1:s+1}) \leq \frac{\pi_t(x_{1:t})}{\pi_{s+1}(x_{1:s+1})} \leq \lambda_{s+1,t}^{+} \cdot \pi_t(x_{s+2:t}|x_{1:s+1}),
\end{align*}
where $0 < \lambda_{s+1,t}^-, \lambda_{s+1,t}^+ < \infty$, implies that
\begin{align}
\Psi_{s,t}^M(x_{1:s};\varphi) \convD \frac{\pi_t(x_{1:s})^2}{\pi_{s}(x_{1:s})^2} \left( \int \varphi(x_{1:t}) \pi_t(x_{s+1:t}|x_{1:s})\myd x_{s+1:t} - \pi_t(\varphi)\right)^2, ~\text{as}~M \to \infty.
\end{align}
\end{lemma}
\begin{proof}
% IMPROVE FOR SUBMISSION
Under the strong mixing assumption and given that we use a \smc method to generate properly weighted samples the result follows from standard \smc results \citep{DelMoral:2004}.
\end{proof}

\begin{theorem}[Vitali Convergence Theorem]
\label{thm:vct}
If $\{\Psi_{s,t}^M(x_{1:s};\varphi)\}$ is uniformly integrable and if $\Psi_{s,t}^M(x_{1:s};\varphi) \convD \Psi_{s,t}(x_{1:s};\varphi)$, then
\begin{align*}
\lim_{M\to\infty}\int \Psi_{s,t}^M(x_{1:s};\varphi) \pi_s(x_{1:s}) \myd x_{1:s} = \int \Psi_{s,t}(x_{1:s};\varphi) \pi_s(x_{1:s}) \myd x_{1:s}.
\end{align*}
\end{theorem}
\begin{proof}
See \citet[Chapter 6]{folland1999}.% \cn{Actually just a statement of the theorem... as an exercise...}
\end{proof}

Under assumptions of uniform integrability and strong mixing (or Assumption~\ref{ass:conv}), the result now follows by using the Vitali convergence theorem~\ref{thm:vct} and noting that
\begin{align*}
\int \Psi_{s,t}(x_{1:s};\varphi) \pi_s(x_{1:s}) \myd x_{1:s} = \int \frac{\pi_t(x_{1:s})^2}{\pi_{s}(x_{1:s})} \left( \int \varphi(x_{1:t}) \pi_t(x_{s+1:t}|x_{1:s})\myd x_{s+1:t} - \pi_t(\varphi)\right)^2 \myd x_{1:s}.
\end{align*}

\newpage
% ======================================================================
%                           N vs M
% ======================================================================
\subsubsection*{Choosing $N$ and $M$}
The constants in Proposition~\ref{prop:nvsm} in the main manuscript are defined as follows
\begin{align*}
A_t &= \int x_{t,d}^2 \pi_t(x_{1:t,d}) \myd x_{1:t,d}, \\
A_s &= \int \frac{\pi_t(x_{1:s,d})^2}{\pi_s(x_{1:s,d})} \left(\int x_{t,d} \pi_t(x_{t,d}|x_{s,d}) \myd x_{t,d} \right)^2 \myd x_{1:s,d}, \\
\tilde A_s &= \int \frac{\pi_t(x_{1:s+1,d})^2}{\pi_s(x_{1:s,d}) r(x_{s+1,d}|x_{s,d})} \left(\int x_{t,d} \pi_t(x_{t,d}|x_{s+1,d}) \myd x_{t,d} \right)^2 \myd x_{1:s+1,d},\\
B_s &= \int \frac{\pi_t(x_{1:s,d})^2}{\pi_s(x_{1:s,d})} \myd x_{1:s,d}, \quad\quad \tilde B_s = \int \frac{\pi_t(x_{1:s+1,d})^2}{\pi_s(x_{1:s,d}) r(x_{s+1,d}|x_{s,d})} \myd x_{1:s+1,d},\\
C_s &= \int \frac{\pi_t(x_{1:s,d})^2}{\pi_s(x_{1:s,d})} \int x_{t,d} \pi_t(x_{t,d}|x_{s,d}) \myd x_{t,d} \myd x_{1:s,d},\\
\tilde C_s &= \int \frac{\pi_t(x_{1:s+1,d})^2}{\pi_s(x_{1:s,d}) r(x_{s+1,d}|x_{s,d})} \int x_{t,d} \pi_t(x_{t,d}|x_{s+1,d}) \myd x_{t,d} \myd x_{1:s+1,d},
\end{align*}
with $A_0 = 0, B_0 = 1, C_0 = 0$, 
\begin{align*}
\tilde A_0 &= \int \frac{\pi_t(x_{1,d})^2}{r(x_{1,d})} \left(\int x_{t,d} \pi_t(x_{t,d}|x_{1,d}) \myd x_{t,d} \right)^2 \myd x_{1,d},\\
\tilde B_0 &= \int \frac{\pi_t(x_{1,d})^2}{r(x_{1,d})} \myd x_{1,d},\\
\tilde C_0 &= \int \frac{\pi_t(x_{1,d})^2}{r(x_{1,d})} \int x_{t,d} \pi_t(x_{t,d}|x_{1,d}) \myd x_{t,d} \myd x_{1,d},
\end{align*}
and 
\begin{align*}
\tilde A_{t-1} &= \int \frac{\pi_t(x_{1:t,d})^2}{\pi_s(x_{1:t-1,d}) r(x_{t,d}|x_{t-1,d})} x_{t,d}^2 \myd x_{1:t,d},\\
\tilde B_{t-1} &= \int \frac{\pi_t(x_{1:t,d})^2}{\pi_s(x_{1:t-1,d}) r(x_{t,d}|x_{t-1,d})} \myd x_{1:t,d},\\
\tilde C_{t-1} &= \int \frac{\pi_t(x_{1:t,d})^2}{\pi_s(x_{1:t-1,d}) r(x_{t,d}|x_{t-1,d})} x_{t,d} \myd x_{1:t,d}.
\end{align*}

\subsubsection*{Proof of Proposition~\ref{prop:nvsm}}
For fully adapted \smc we have from the result in \citet{johansen2008} (see also our convergence result in the previous section) and for the model defined in the main manuscript
\begin{align*}
&\pi_t((\varphi-\pi_t(\varphi))^2) = n_x A_t\\
&\int \frac{\pi_{t}(x_{1:s})^2}{\pi_s(x_{1:s})} \left(\sum_{d=1}^{n_x}\int  x_{t,d} \pi_t(x_{s+1:t}|x_{1:s}) \myd x_{s+1:t} \right)^2 \myd x_{1:s} = \\
&=\sum_{e=1}^{n_x} \sum_{f=1}^{n_x} \int \left[ \frac{\pi_{t}(x_{1:s})^2}{\pi_s(x_{1:s})} \int  x_{t,e} \pi_t(x_{s+1:t}|x_{1:s}) \myd x_{s+1:t} \cdot \int  x_{t,f} \pi_t(x_{s+1:t}|x_{1:s}) \myd x_{s+1:t} \right]\myd x_{1:s}\\
&= \sum_{e=1}^{n_x} \sum_{f=1}^{n_x} \int \left[ \frac{\pi_{t}(x_{1:s})^2}{\pi_s(x_{1:s})} \int  x_{t,e} \pi_t(x_{t,e}|x_{s,e}) \myd x_{t,e} \cdot \int  x_{t,f} \pi_t(x_{t,f}|x_{s,f}) \myd x_{t,f} \right]\myd x_{1:s} \\
&= n_x B_s^{n_x-1} A_s + n_x(n_x-1) B_s^{n_x-2} C_s^2,
\end{align*}
with constants as defined above.

For nested \smc we have $r(x_s|x_{s-1}) = \prod_{d=1}^{n_x} r(x_{s,d}|x_{s-1,d})$ and due to the independence between dimensions we will have no dependence on internal ancestor variables in $\eta_{s}, 
\kappa_{s+1}, \tau_{s+1}$, \ie
\begin{align*}
\eta_s(u_s|x_{1:s}) &= \prod_{d=1}^{n_x} \prod_{j=1}^M r(x_{s+1,d}^j|x_{s,d}),\\
\kappa_{s+1}(x_{s+1}|u_s) &= \prod_{d=1}^{n_x} \sum_{j=1}^M \frac{w_d^j}{\sum_\ell w_d^\ell} \delta_{x_{s+1,d}^j}(\myd x_{s+1,d}),\\
\tau_{s+1}(u_s) &= \prod_{d=1}^{n_x} \frac{1}{M} \sum_{j=1}^M w_d^j,\\
w_d^j &= \frac{f(x_{s+1,d}^j|x_{s,d})g(y_{s+1,d}|x_{s+1,d}^j)}{r(x_{s+1,d}^j|x_{s,d})}.
\end{align*}
For the variance contribution of the final step we obtain $\pi_t((\varphi-\pi_t(\varphi))^2) = n_x \sigma_x^2$, the same result as fully adapted \smc. The remaining can be calculated as follows
\begin{align}
&\int \Bigg[ \frac{Z_s^2 \tau_{s+1}(u_s)^2}{Z_{s+1}^2} \left( \int \varphi(x_{1:t}) \frac{\pi_t(x_{1:t})}{\pi_{s+1}(x_{1:s+1})} \kappa_{s+1}^M(x_{s+1}|u_s) \myd x_{s+1:t} \right)^2 \eta_s^M(u_s|x_{1:s-1})\pi_s(x_{1:s}) \Bigg] \myd u_s \myd x_{1:s} \nonumber\\
&= \frac{1}{p(y_{s+1}|y_{1:s})^2} \int \Bigg[ \tau_{s+1}(u_s)^2 \frac{1}{M^{2n_x} \tau_{s+1}(u_s)^2}
\eta_s^M(u_s|x_{1:s-1})\pi_s(x_{1:s}) \nonumber\\
& \quad \cdot \left( \sum_{e=1}^{n_x} \left[\sum_{j=1}^M w_e^j \frac{\int x_{t,e} \pi_t(x_{1:s,e},x_{s+1,e}^j,x_{t,e}) \myd x_{t,e}}{\pi_{s+1}(x_{1:s,e},x_{s+1,e}^j)} \cdot \prod_{d\neq e} \sum_{j=1}^M w_d^j \frac{\pi_t(x_{1:s,d},x_{s+1,d}^j)}{\pi_{s+1}(x_{1:s,d},x_{s+1,d}^j)}\right]\right)^2 \Bigg] \myd u_s \myd x_{1:s} \nonumber\\
&=\frac{1}{p(y_{s+1}|y_{1:s})^2 M^{2n_x}} \sum_{e=1}^{n_x}\sum_{e'=1}^{n_x} \int \tilde h_e(x_{1:s},u_s) \tilde h_{e'}(x_{1:s},u_s) \prod_{d=1}^{n_x} \left[ \pi_s(x_{1:s,d}) \prod_{j=1}^M r(x_{s+1,d}^j|x_{s,d})\right]\myd u_s \myd x_{1:s},\label{eq:indep:incvar}
\end{align}
for $\tilde h_e$ defined by
\begin{align*}
\tilde h_e(x_{1:s},u_s) &= \sum_{j=1}^M w_e^j \frac{\int x_{t,e} \pi_t(x_{1:s,e},x_{s+1,e}^j,x_{t,e}) \myd x_{t,e}}{\pi_{s+1}(x_{1:s,e},x_{s+1,e}^j)} \cdot \prod_{d\neq e} \sum_{j=1}^M w_d^j \frac{\pi_t(x_{1:s,d},x_{s+1,d}^j)}{\pi_{s+1}(x_{1:s,d},x_{s+1,d}^j)}.
\end{align*}
Now, note that
\begin{align*}
\tilde h_e(x_{1:s},u_s)^2 &= \sum_{i_{1:n_x},j_{1:n_x}} \Bigg[ \prod_{d=1}^{n_x} w_d^{i_d} w_d^{j_d} \cdot \prod_{d \neq e} \frac{\pi_t(x_{1:s,d},x_{s+1,d}^{i_d})}{\pi_{s+1}(x_{1:s,d},x_{s+1,d}^{i_d})} \frac{\pi_t(x_{1:s,d},x_{s+1,d}^{j_d})}{\pi_{s+1}(x_{1:s,d},x_{s+1,d}^{j_d})} \\
&\quad \cdot \frac{\int x_{t,e} \pi_t(x_{1:s,e},x_{s+1,e}^{i_e},x_{t,e}) \myd x_{t,e}}{\pi_{s+1}(x_{1:s,e},x_{s+1,e}^{i_e})} \frac{\int x_{t,e} \pi_t(x_{1:s,e},x_{s+1,e}^{j_e},x_{t,e}) \myd x_{t,e}}{\pi_{s+1}(x_{1:s,e},x_{s+1,e}^{j_e})}\Bigg],\\
\tilde h_e(x_{1:s},u_s) \tilde h_{e'}(x_{1:s},u_s) &= \sum_{i_{1:n_x},j_{1:n_x}} \Bigg[ \prod_{d=1}^{n_x} w_d^{i_d} w_d^{j_d} \cdot \prod_{d \neq e} \frac{\pi_t(x_{1:s,d},x_{s+1,d}^{i_d})}{\pi_{s+1}(x_{1:s,d},x_{s+1,d}^{i_d})} \prod_{d \neq e'} \frac{\pi_t(x_{1:s,d},x_{s+1,d}^{j_d})}{\pi_{s+1}(x_{1:s,d},x_{s+1,d}^{j_d})} \\
&\quad \cdot \frac{\int x_{t,e} \pi_t(x_{1:s,e},x_{s+1,e}^{i_e},x_{t,e}) \myd x_{t,e}}{\pi_{s+1}(x_{1:s,e},x_{s+1,e}^{i_e})} \frac{\int x_{t,e'} \pi_t(x_{1:s,e'},x_{s+1,e'}^{j_{e'}},x_{t,e'}) \myd x_{t,e'}}{\pi_{s+1}(x_{1:s,e'},x_{s+1,e'}^{j_{e'}})}\Bigg],
\end{align*}
with $e \neq e'$ and all $i_d,j_d \in \{1,\ldots,M\}$. 

We will in the sequel also make use of the following observation
\begin{align}
\frac{w_d^i}{\pi_{s+1}(x_{1:s,d},x_{s+1,d}^{i})} = \frac{p(y_{s+1,d}|y_{1:s,d})}{r(x_{s+1,d}^i|x_{s,d}) \pi_s(x_{1:s,d})}.\label{eq:weightsident}
\end{align}

Now, we consider the case in \eqref{eq:indep:incvar} when $e=e'$:
\begin{align*}
&\frac{1}{p(y_{s+1}|y_{1:s})^2 M^{2n_x}}\int \tilde h_e(x_{1:s},u_s)^2 \prod_{d=1}^{n_x} \left[ \pi_s(x_{1:s,d}) \prod_{j=1}^M r(x_{s+1,d}^j|x_{s,d})\right]\myd u_s \myd x_{1:s}\\
&= \frac{1}{M^{2n_x}} \sum_{i_{1:n_x},j_{1:n_x}} \Bigg[ \prod_{d\neq e} \int \frac{\prod_{j=1}^M r(x_{s+1,d}^j|x_{s,d}) \pi_t(x_{1:s,d},x_{s+1,d}^{i_d}) \pi_t(x_{1:s,d},x_{s+1,d}^{j_d})}{r(x_{s+1,d}^{i_d}|x_{s,d}) r(x_{s+1,d}^{j_d}|x_{s,d}) \pi_s(x_{1:s,d})} \myd u_{s,d} \myd x_{1:s,d} \\
&\cdot \int \Big[ \frac{\prod_{j=1}^M r(x_{s+1,e}^j|x_{s,e}) \pi_t(x_{1:s,e})^2}{r(x_{s+1,e}^{i_e}|x_{s,e}) r(x_{s+1,e}^{j_e}|x_{s,e}) \pi_s(x_{1:s,e})} \\
&\quad\quad\int x_{t,e} \pi_t(x_{s+1,e}^{i_e},x_{t,e}|x_{s,e})\myd x_{t,e} \int x_{t,e} \pi_t(x_{s+1,e}^{j_e},x_{t,e}|x_{s,e})\myd x_{t,e} \Big]  \myd u_{s,e} \myd x_{1:s,e} \Bigg]\\
&= B_s^{n_x-1} \left( A_s + M^{-1}\left(\tilde A_s - A_s \right) \right) \left(1 - \frac{1}{M} \right)^{n_x-1}\left(1 + \frac{\tilde B_s}{B_s (M-1)} \right)^{n_x-1},
\end{align*}
where in the first equality we have used \eqref{eq:weightsident} and independency over dimensions. The second equality follows by straightforward (but tedious) calculations using combinatorial identities and noting that by definition of the model the constants do not depend on the dimension~$d$.

Let us now consider the case in \eqref{eq:indep:incvar} when $e\neq e'$:
\begin{align*}
&\frac{1}{p(y_{s+1}|y_{1:s})^2 M^{2n_x}}\int \tilde h_e(x_{1:s},u_s) \tilde h_{e'}(x_{1:s},u_s)  \prod_{d=1}^{n_x} \left[ \pi_s(x_{1:s,d}) \prod_{j=1}^M r(x_{s+1,d}^j|x_{s,d})\right]\myd u_s \myd x_{1:s}\\
&= \frac{1}{M^{2n_x}} \sum_{i_{1:n_x},j_{1:n_x}} \Bigg[ \prod_{d\neq e,e'} \int \frac{\prod_{j=1}^M r(x_{s+1,d}^j|x_{s,d}) \pi_t(x_{1:s,d},x_{s+1,d}^{i_d}) \pi_t(x_{1:s,d},x_{s+1,d}^{j_d})}{r(x_{s+1,d}^{i_d}|x_{s,d}) r(x_{s+1,d}^{j_d}|x_{s,d}) \pi_s(x_{1:s,d})} \myd u_{s,d} \myd x_{1:s,d} \\
&\cdot \int \frac{\pi_t(x_{1:s,e},x_{s+1,e}^{j_e}) \prod_{j=1}^M r(x_{s+1,e}^j|x_{s,e}) }{r(x_{s+1,e}^{i_e}|x_{s,e}) r(x_{s+1,e}^{j_e}|x_{s,e}) \pi_s(x_{1:s,e})} \int x_{t,e} \pi_t(x_{1:s,e},x_{s+1,e}^{i_e},x_{t,e})\myd x_{t,e} \myd u_{s,e} \myd x_{1:s,e} \\
&\cdot  \int \frac{ \pi_t(x_{1:s,e'},x_{s+1,e'}^{i_{e'}})\prod_{j=1}^M r(x_{s+1,e'}^j|x_{s,e'})}{r(x_{s+1,e'}^{i_{e'}}|x_{s,e'}) r(x_{s+1,e'}^{j_{e'}}|x_{s,e'}) \pi_s(x_{1:s,e'})}\int x_{t,e'} \pi_t(x_{1:s,e'},x_{s+1,e'}^{j_e},x_{t,e'})\myd x_{t,e'}  \myd u_{s,e'} \myd x_{1:s,e'} \Bigg]\\
&= \left(C_s + M^{-1}\left( \tilde C_s - C_s\right) \right)^2  \left(1 - \frac{1}{M} \right)^{n_x-2}\left(1 + \frac{\tilde B_s}{B_s (M-1)} \right)^{n_x-2},
\end{align*}
where again we have made use of indepency over dimensions $d$ and \eqref{eq:weightsident}. The last equality follows again by straightforward manipulations and we can see that product $\prod_{d\neq e,e'} \cdot$ is more or less equal to the one above, hence we obtain $B_s^{n_x-2}$ instead of $B_s^{n_x-1}$.

Putting all this together we get that
\begin{align*}
\Sigma_t^M(\varphi) &= n_x A_t + \sum_{s=0}^{t-1} \Bigg[ n_x B_s^{n_x-1} \left( A_s + M^{-1}\left(\tilde A_s - A_s \right) \right) \left(1 - \frac{1}{M} \right)^{n_x-1}\left(1 + \frac{\tilde B_s}{B_s (M-1)} \right)^{n_x-1} \\
& +n_x(n_x-1) B_s^{n_x-2} \left(C_s + M^{-1}\left( \tilde C_s - C_s\right) \right)^2  \left(1 - \frac{1}{M} \right)^{n_x-2}\left(1 + \frac{\tilde B_s}{B_s (M-1)} \right)^{n_x-2} \Bigg],
\end{align*}
equality follows by noting that $\sum_{e,e'}  = \sum_{e}\sum_{e'=e} + \sum_{e}\sum_{e'\neq e}$ and that the constants do not depend on $e/e'$.

\newpage
% ======================================================================
%                           Experiments
% ======================================================================
\subsection{Comparison with Independent Resampling Particle Filter}\label{supp:exp}

% ======================================================================
%                           IR-PF
% ======================================================================
We compare several variants of \nsmc to Independent Resampling Particle Filter (IR-PF) on the same setup studied in \citet[High dimensional problems]{lamberti2016independent}, for more information on the model and setup we refer to that paper. Figure~\ref{fig:IRPF} illustrates the results for  $N=M \in \{10,100\}$ and as we can see \nsmc outperforms IR-PF significantly in root mean square error (RMSE). \nsmc-\is and \nsmc-PF both approximate the optimal proposal \smc and as such generate conditionally independent samples (see supplementary methods section above for how to use \is  as a nested procedure). \nsmc-FAPF, clearly the best of all of them, on the other hand, approximates the fully adapted \smc and generates conditionally \emph{dependent} samples.
\begin{figure}
    \centering
    \begin{subfigure}[b]{0.7\textwidth}
        \includegraphics[width=\textwidth]{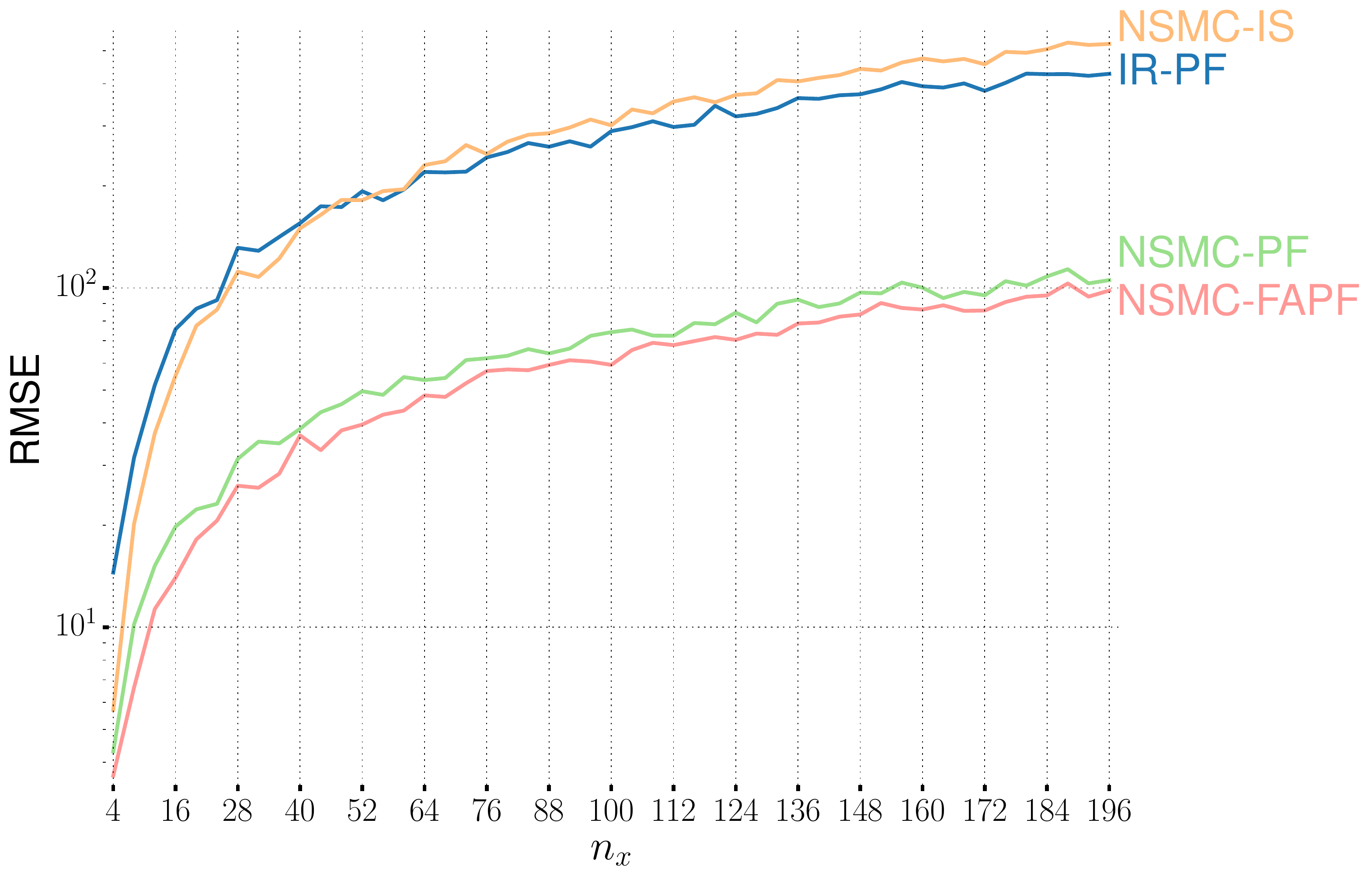}
        \caption{$N=M=10$}
        \label{fig:N10M10}
    \end{subfigure}
    
    \begin{subfigure}[b]{0.7\textwidth}
        \includegraphics[width=\textwidth]{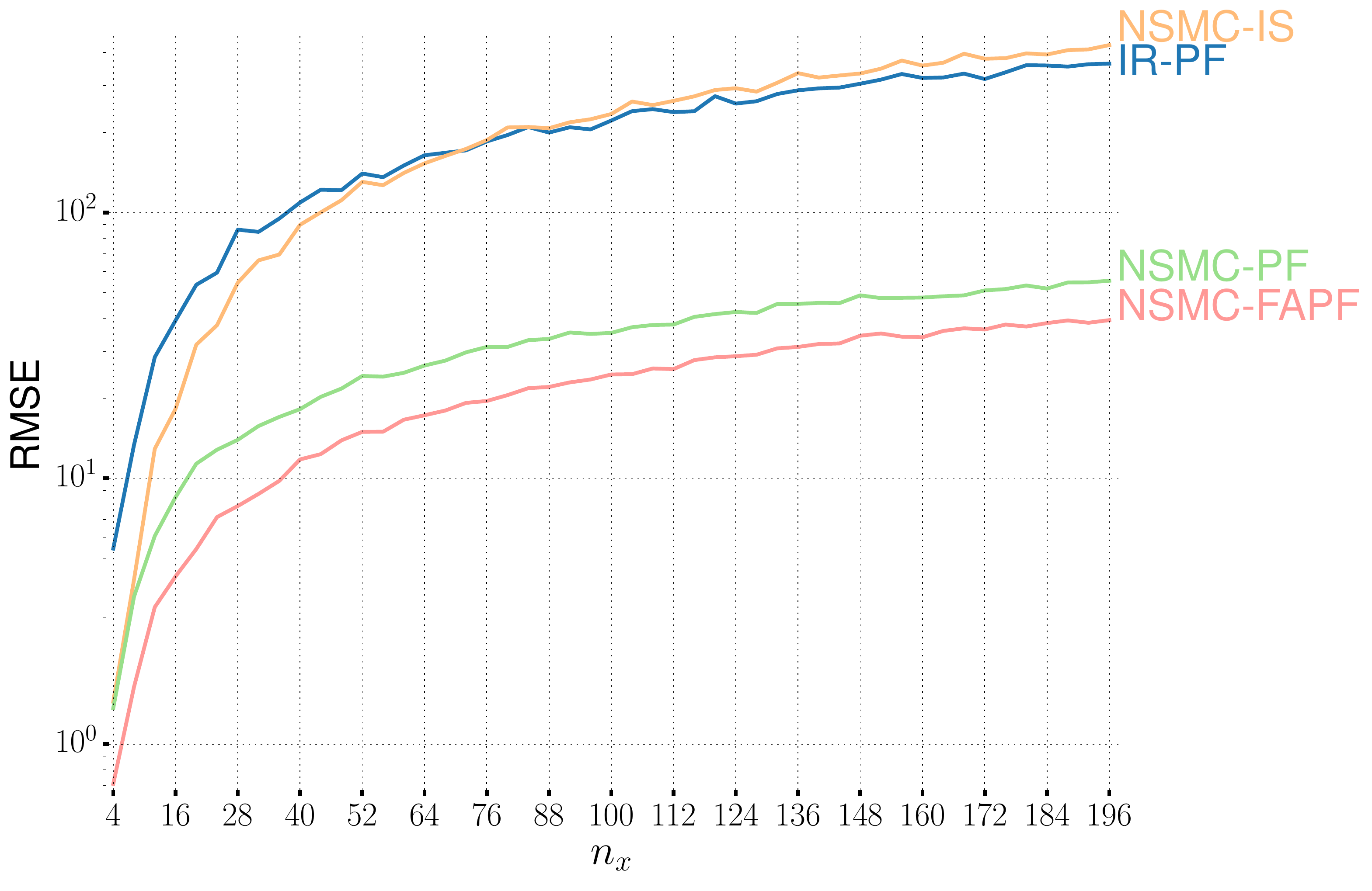}
        \caption{$N=M=100$}
        \label{fig:N100M100}
    \end{subfigure}
    \caption{RMSE of the IR-PF and three types of \nsmc methods, approximation of optimal proposal \smc using IS (orange) and PF (green), approximation of fully adapted \smc using PF with BS (red).}\label{fig:IRPF}
\end{figure}

\bibliographystyle{abbrvnat}
\bibliography{references}

\begin{thebibliography}{55}
\providecommand{\natexlab}[1]{#1}
\providecommand{\url}[1]{\texttt{#1}}
\expandafter\ifx\csname urlstyle\endcsname\relax
  \providecommand{\doi}[1]{doi: #1}\else
  \providecommand{\doi}{doi: \begingroup \urlstyle{rm}\Url}\fi

\bibitem[Andrieu et~al.(2010)Andrieu, Doucet, and
  Holenstein]{andrieuDH2010particle}
C.~Andrieu, A.~Doucet, and R.~Holenstein.
\newblock Particle {M}arkov chain {M}onte {C}arlo methods.
\newblock \emph{Journal of the Royal Statistical Society: Series B (Statistical
  Methodology)}, 72\penalty0 (3):\penalty0 269--342, 2010.

\bibitem[Beskos et~al.(2014)Beskos, Crisan, Jasra, Kamatani, and
  Zhou]{beskosCJKZ2014a}
A.~Beskos, D.~Crisan, A.~Jasra, K.~Kamatani, and Y.~Zhou.
\newblock A stable particle filter in high-dimensions.
\newblock \emph{ArXiv:1412.3501}, Dec. 2014.

\bibitem[Briggs et~al.(2013)Briggs, Dowd, and Meyer]{briggsDM2013data}
J.~Briggs, M.~Dowd, and R.~Meyer.
\newblock Data assimilation for large-scale spatio-temporal systems using a
  location particle smoother.
\newblock \emph{Environmetrics}, 24\penalty0 (2):\penalty0 81--97, 2013.

\bibitem[Capp\'e et~al.(2005)Capp\'e, Moulines, and Ryd\'en]{cappe2005}
O.~Capp\'e, E.~Moulines, and T.~Ryd\'en.
\newblock \emph{{Inference in Hidden Markov Models}}.
\newblock Springer-Verlag New York, 2005.

\bibitem[Carpenter et~al.(1999)Carpenter, Clifford, and
  Fearnhead]{CarpenterCF:1999}
J.~Carpenter, P.~Clifford, and P.~Fearnhead.
\newblock Improved particle filter for nonlinear problems.
\newblock \emph{IEE Proceedings Radar, Sonar and Navigation}, 146\penalty0
  (1):\penalty0 2--7, 1999.

\bibitem[Carter and Kohn(1994)]{carter1994gibbs}
C.~K. Carter and R.~Kohn.
\newblock On {G}ibbs sampling for state space models.
\newblock \emph{Biometrika}, 81\penalty0 (3):\penalty0 541--553, 1994.

\bibitem[Chen et~al.(2011)Chen, Sch\"on, Ohlsson, and
  Ljung]{chenSOL2011decentralized}
T.~Chen, T.~B. Sch\"on, H.~Ohlsson, and L.~Ljung.
\newblock Decentralized particle filter with arbitrary state decomposition.
\newblock \emph{IEEE Transactions on Signal Processing}, 59\penalty0
  (2):\penalty0 465--478, Feb 2011.

\bibitem[Chopin et~al.(2013)Chopin, Jacob, and
  Papaspiliopoulos]{chopinJP2013smc2}
N.~Chopin, P.~E. Jacob, and O.~Papaspiliopoulos.
\newblock {SMC}2: an efficient algorithm for sequential analysis of state space
  models.
\newblock \emph{Journal of the Royal Statistical Society: Series B (Statistical
  Methodology)}, 75\penalty0 (3):\penalty0 397--426, 2013.

\bibitem[Clifford et~al.(2014)Clifford, Pagendam, Baldock, Cressie,
  Farquharson, Farrell, Macdonald, and Murray]{clifford2014}
D.~Clifford, D.~Pagendam, J.~Baldock, N.~Cressie, R.~Farquharson, M.~Farrell,
  L.~Macdonald, and L.~Murray.
\newblock Rethinking soil carbon modelling: a stochastic approach to quantify
  uncertainties.
\newblock \emph{Environmetrics}, 25\penalty0 (4):\penalty0 265--278, 2014.

\bibitem[Cohen(2004)]{cohen2004bioinformatics}
J.~Cohen.
\newblock Bioinformatics—an introduction for computer scientists.
\newblock \emph{ACM Computing Surveys (CSUR)}, 36\penalty0 (2):\penalty0
  122--158, 2004.

\bibitem[Cressie and Wikle(2011)]{CressieW:2011}
N.~Cressie and C.~K. Wikle.
\newblock \emph{Statistics for spatio-temporal data}.
\newblock Wiley, 2011.

\bibitem[Del~Moral(2004)]{DelMoral:2004}
P.~Del~Moral.
\newblock \emph{{F}eynman-{K}ac Formulae - Genealogical and Interacting
  Particle Systems with Applications}.
\newblock Probability and its Applications. Springer-Verlag New York, 2004.

\bibitem[Djuric and Bugallo(2013)]{djuric2013particle}
P.~M. Djuric and M.~F. Bugallo.
\newblock Particle filtering for high-dimensional systems.
\newblock In \emph{Computational Advances in Multi-Sensor Adaptive Processing
  (CAMSAP), 2013 IEEE 5th International Workshop on}, pages 352--355. IEEE,
  2013.

\bibitem[Douc et~al.(2014)Douc, Moulines, and Stoffer]{douc2014nonlinear}
R.~Douc, E.~Moulines, and D.~Stoffer.
\newblock \emph{Nonlinear time series: Theory, methods and applications with R
  examples}.
\newblock CRC Press, 2014.

\bibitem[Everitt(2012)]{everitt2012bayesian}
R.~G. Everitt.
\newblock Bayesian parameter estimation for latent {M}arkov random fields and
  social networks.
\newblock \emph{Journal of Computational and Graphical Statistics}, 21\penalty0
  (4):\penalty0 940--960, 2012.

\bibitem[Fearnhead and Clifford(2003)]{fearnhead2003line}
P.~Fearnhead and P.~Clifford.
\newblock On-line inference for hidden {M}arkov models via particle filters.
\newblock \emph{Journal of the Royal Statistical Society: Series B (Statistical
  Methodology)}, 65\penalty0 (4):\penalty0 887--899, 2003.

\bibitem[Fearnhead et~al.(2010)Fearnhead, Papaspiliopoulos, Roberts, and
  Stuart]{fearnheadPRS2010random}
P.~Fearnhead, O.~Papaspiliopoulos, G.~O. Roberts, and A.~Stuart.
\newblock Random-weight particle filtering of continuous time processes.
\newblock \emph{Journal of the Royal Statistical Society: Series B (Statistical
  Methodology)}, 72\penalty0 (4):\penalty0 497--512, 2010.

\bibitem[Folland(1999)]{folland1999}
G.~B. Folland.
\newblock \emph{Real analysis}.
\newblock Pure and Applied Mathematics (New York). John Wiley \& Sons, Inc.,
  New York, second edition, 1999.
\newblock Modern techniques and their applications.

\bibitem[Fr\"uhwirth-Schnatter(1994)]{fruhwirth1994data}
S.~Fr\"uhwirth-Schnatter.
\newblock Data augmentation and dynamic linear models.
\newblock \emph{Journal of Time Series Analysis}, 15\penalty0 (2):\penalty0
  183--202, 1994.

\bibitem[Fu et~al.(2012)Fu, Banerjee, Liess, and Snyder]{fuBLS2012drought}
Q.~Fu, A.~Banerjee, S.~Liess, and P.~K. Snyder.
\newblock Drought detection of the last century: An {MRF}-based approach.
\newblock In \emph{Proceedings of the 2012 SIAM International Conference on
  Data Mining}, pages 24--34, Anaheim, CA, USA, April 2012.

\bibitem[Godsill et~al.(2004)Godsill, Doucet, and West]{GodsillDW:2004}
S.~J. Godsill, A.~Doucet, and M.~West.
\newblock {M}onte {C}arlo smoothing for nonlinear time series.
\newblock \emph{Journal of the American Statistical Association}, 99\penalty0
  (465):\penalty0 156--168, Mar. 2004.

\bibitem[Gordon et~al.(1993)Gordon, Salmond, and Smith]{GordonSS:1993}
N.~J. Gordon, D.~J. Salmond, and A.~F.~M. Smith.
\newblock Novel approach to nonlinear/non-{G}aussian {B}ayesian state
  estimation.
\newblock \emph{Radar and Signal Processing, {IEE} Proceedings {F}},
  140\penalty0 (2):\penalty0 107 --113, Apr. 1993.

\bibitem[Hamze and de~Freitas(2005)]{hamze2005hot}
F.~Hamze and N.~de~Freitas.
\newblock Hot coupling: a particle approach to inference and normalization on
  pairwise undirected graphs of arbitrary topology.
\newblock In \emph{Advances in Neural Information Processing Systems (NIPS)},
  2005.

\bibitem[Jaoua et~al.(2013)Jaoua, Duflos, Vanheeghe, and
  Septier]{jaoua2013bayesian}
N.~Jaoua, E.~Duflos, P.~Vanheeghe, and F.~Septier.
\newblock Bayesian nonparametric state and impulsive measurement noise density
  estimation in nonlinear dynamic systems.
\newblock In \emph{2013 IEEE International Conference on Acoustics, Speech and
  Signal Processing}, pages 5755--5759, May 2013.

\bibitem[Johansen and Doucet(2008)]{johansen2008}
A.~M. Johansen and A.~Doucet.
\newblock A note on auxiliary particle filters.
\newblock \emph{Statistics \& Probability Letters}, 78\penalty0 (12):\penalty0
  1498 -- 1504, 2008.

\bibitem[Johansen et~al.(2012)Johansen, Whiteley, and
  Doucet]{johansenWD2012exact}
A.~M. Johansen, N.~Whiteley, and A.~Doucet.
\newblock Exact approximation of {R}ao-{B}lackwellised particle filters.
\newblock In \emph{Proceesings of the 16th IFAC Symposium on System
  Identification (SYSID)}, pages 488--493, Brussels, Belgium, 2012.

\bibitem[Jordan(2004)]{jordan2004graphical}
M.~I. Jordan.
\newblock Graphical models.
\newblock \emph{Statistical Science}, 19\penalty0 (1):\penalty0 140--155, 2004.

\bibitem[Kalman(1960)]{Kalman:1960}
R.~E. Kalman.
\newblock A new approach to linear filtering and prediction problems.
\newblock \emph{Transactions of the {ASME}, Journal of Basic Engineering},
  82:\penalty0 35--45, 1960.

\bibitem[Kitagawa(1996)]{kitagawa1996monte}
G.~Kitagawa.
\newblock Monte {C}arlo filter and smoother for non-{G}aussian nonlinear state
  space models.
\newblock \emph{Journal of computational and graphical statistics}, 5\penalty0
  (1):\penalty0 1--25, 1996.

\bibitem[{Lamberti} et~al.(2016){Lamberti}, {Petetin}, {Desbouvries}, and
  {Septier}]{lamberti2016independent}
R.~{Lamberti}, Y.~{Petetin}, F.~{Desbouvries}, and F.~{Septier}.
\newblock {Independent Resampling Sequential Monte Carlo Algorithms}.
\newblock \emph{ArXiv e-prints}, July 2016.

\bibitem[Lindsten and Sch\"on(2013)]{LindstenS:2013}
F.~Lindsten and T.~B. Sch\"on.
\newblock Backward simulation methods for {M}onte {C}arlo statistical
  inference.
\newblock \emph{Foundations and Trends in Machine Learning}, 6\penalty0
  (1):\penalty0 1--143, 2013.

\bibitem[Lindsten et~al.(2016)Lindsten, Johansen, Naesseth, Kirkpatrick,
  Sch\"{o}n, Aston, and Bouchard-C\^ot\'e]{lindstenjnksab2014}
F.~Lindsten, A.~M. Johansen, C.~A. Naesseth, B.~Kirkpatrick, T.~B. Sch\"{o}n,
  J.~Aston, and A.~Bouchard-C\^ot\'e.
\newblock Divide-and-conquer with sequential {M}onte {C}arlo.
\newblock \emph{Journal of Computational and Graphical Statistics}, 2016.
\newblock (accepted for publication).

\bibitem[Martino et~al.(2016)Martino, Elvira, and
  Louzada]{martino2016weighting}
L.~Martino, V.~Elvira, and F.~Louzada.
\newblock Weighting a resampled particles in sequential monte carlo (extended
  preprint).
\newblock \emph{viXra e-prints}, Feb. 2016.

\bibitem[Monteleoni et~al.(2013)Monteleoni, Schmidt, Alexander,
  Niculescu-Mizil, Steinhaeuser, Tippett, Banerjee, Blumenthal, Auroop
  R.~Ganguly, and Tedesco]{monteleoniEtAl2013}
C.~Monteleoni, G.~A. Schmidt, F.~Alexander, A.~Niculescu-Mizil,
  K.~Steinhaeuser, M.~Tippett, A.~Banerjee, M.~B. Blumenthal, J.~E.~S. Auroop
  R.~Ganguly, and M.~Tedesco.
\newblock Climate informatics.
\newblock In T.~Yu, N.~Chawla, and S.~Simoff, editors, \emph{Computational
  Intelligent Data Analysis for Sustainable Development}. Chapman and Hall/CRC,
  London, 2013.

\bibitem[Murray(2016)]{murray2016}
L.~Murray.
\newblock Personal communication, 2016.

\bibitem[Naesseth et~al.(2014{\natexlab{a}})Naesseth, Lindsten, and
  Sch\"on]{NaessethLS:2014IT}
C.~A. Naesseth, F.~Lindsten, and T.~B. Sch\"on.
\newblock Capacity estimation of two-dimensional channels using sequential
  {M}onte {C}arlo.
\newblock In \emph{Proceedings of the IEEE Information Theory Workshop
  ({ITW})}, Hobart, Tasmania, Australia, November 2014{\natexlab{a}}.

\bibitem[Naesseth et~al.(2014{\natexlab{b}})Naesseth, Lindsten, and
  Sch\"{o}n]{naessethLS2014}
C.~A. Naesseth, F.~Lindsten, and T.~B. Sch\"{o}n.
\newblock Sequential {M}onte {C}arlo for {G}raphical {M}odels.
\newblock In \emph{Advances in Neural Information Processing Systems 27}, pages
  1862--1870. Curran Associates, Inc., Montreal, Canada, 2014{\natexlab{b}}.

\bibitem[Naesseth et~al.(2015{\natexlab{a}})Naesseth, Lindsten, and
  Sch\"{o}n]{naessethLS2015nested}
C.~A. Naesseth, F.~Lindsten, and T.~B. Sch\"{o}n.
\newblock Nested sequential {M}onte {C}arlo methods.
\newblock In \emph{The 32nd International Conference on Machine Learning},
  volume~37 of \emph{JMLR W\&CP}, pages 1292--1301, Lille, France, jul
  2015{\natexlab{a}}.

\bibitem[Naesseth et~al.(2015{\natexlab{b}})Naesseth, Lindsten, and
  Sch\"{o}n]{naessethLS2015nestedARXIV}
C.~A. Naesseth, F.~Lindsten, and T.~B. Sch\"{o}n.
\newblock Nested sequential {M}onte {C}arlo methods.
\newblock \emph{Arxiv pre-print, arXiv:1502.02536v3}, 2015{\natexlab{b}}.

\bibitem[Naesseth et~al.(2015{\natexlab{c}})Naesseth, Lindsten, and
  Sch\"{o}n]{naessethLS2015ws}
C.~A. Naesseth, F.~Lindsten, and T.~B. Sch\"{o}n.
\newblock Towards automated sequential {M}onte {C}arlo for probabilistic
  graphical models.
\newblock In \emph{NIPS Workshop on Black Box Inference and Learning}.
  Montreal, Canada, 2015{\natexlab{c}}.

\bibitem[Pitt and Shephard(1999)]{pittS1999filtering}
M.~K. Pitt and N.~Shephard.
\newblock Filtering via simulation: Auxiliary particle filters.
\newblock \emph{Journal of the American statistical association}, 94\penalty0
  (446):\penalty0 590--599, 1999.

\bibitem[Rebeschini and van Handel(2015)]{rebeschini2015}
P.~Rebeschini and R.~van Handel.
\newblock Can local particle filters beat the curse of dimensionality?
\newblock \emph{Ann. Appl. Probab.}, 25\penalty0 (5):\penalty0 2809--2866, 10
  2015.

\bibitem[{Rebeschini} and {van Handel}(2015)]{rebeschiniH2015can}
P.~{Rebeschini} and R.~{van Handel}.
\newblock {Can local particle filters beat the curse of dimensionality?}
\newblock \emph{Ann. Appl. Probab. (to appear)}, 2015.

\bibitem[Rue and Held(2005)]{RueH:2005}
H.~Rue and L.~Held.
\newblock \emph{Gaussian Markov Random Fields, Theory and Applications}.
\newblock CDC Press, Boca Raton, FL, USA, 2005.

\bibitem[Septier and Peters(2016)]{septier2016}
F.~Septier and G.~W. Peters.
\newblock Langevin and hamiltonian based sequential mcmc for efficient bayesian
  filtering in high-dimensional spaces.
\newblock \emph{IEEE Journal of Selected Topics in Signal Processing},
  10\penalty0 (2):\penalty0 312--327, March 2016.

\bibitem[Shumway and Stoffer(2010)]{shumway2010time}
R.~H. Shumway and D.~S. Stoffer.
\newblock \emph{Time series analysis and its applications: with R examples}.
\newblock Springer Science \& Business Media, 2010.

\bibitem[Snyder et~al.(2015)Snyder, Bengtsson, and
  Morzfeld]{snyder2015performance}
C.~Snyder, T.~Bengtsson, and M.~Morzfeld.
\newblock Performance bounds for particle filters using the optimal proposal.
\newblock \emph{Monthly Weather Review}, 143\penalty0 (11):\penalty0
  4750--4761, 2015.

\bibitem[Stern(2015)]{stern2015}
R.~Stern.
\newblock \emph{A statistical contribution to historical linguistics}.
\newblock PhD thesis, {Carnegie Mellon University}, Department of Statistics,
  Carnegie Mellon University, Pittsburgh PA 15213, 5 2015.

\bibitem[Stewart and McCarty(1992)]{stewart1992}
L.~Stewart and P.~McCarty, Jr.
\newblock Use of {B}ayesian belief networks to fuse continuous and discrete
  information for target recognition, tracking, and situation assessment.
\newblock In \emph{Proc. SPIE}, volume 1699, pages 177--185, 1992.

\bibitem[Tran et~al.(2013)Tran, Scharth, Pitt, and Kohn]{tranSPK2013importance}
M.-N. Tran, M.~Scharth, M.~K. Pitt, and R.~Kohn.
\newblock Importance sampling squared for {B}ayesian inference in latent
  variable models.
\newblock \emph{ArXiv:1309.3339}, sep 2013.

\bibitem[Verg\'e et~al.(2015)Verg\'e, Dubarry, Del~Moral, and
  Moulines]{vergeDDM2013on}
C.~Verg\'e, C.~Dubarry, P.~Del~Moral, and E.~Moulines.
\newblock On parallel implementation of sequential {M}onte {C}arlo methods: the
  island particle model.
\newblock \emph{Statistics and Computing}, 25\penalty0 (2):\penalty0 243--260,
  2015.

\bibitem[Wainwright and Jordan(2008)]{wainwright2008graphical}
M.~J. Wainwright and M.~I. Jordan.
\newblock Graphical models, exponential families, and variational inference.
\newblock \emph{Foundations and Trends{\textregistered} in Machine Learning},
  1\penalty0 (1-2):\penalty0 1--305, 2008.

\bibitem[Wikle(2015)]{Wikle:2015}
C.~K. Wikle.
\newblock Modern perspectives on statistics for spatio-temporal data.
\newblock \emph{WIREs Computational Statistics}, 7\penalty0 (1):\penalty0
  86--98, 2015.

\bibitem[Wikle and Hooten(2010)]{wikle2010general}
C.~K. Wikle and M.~B. Hooten.
\newblock A general science-based framework for dynamical spatio-temporal
  models.
\newblock \emph{Test}, 19\penalty0 (3):\penalty0 417--451, 2010.

\bibitem[{Yang} and {Dunson}(2013)]{yang2013}
Y.~{Yang} and D.~B. {Dunson}.
\newblock {Sequential Markov Chain Monte Carlo}.
\newblock \emph{arXiv:1308.3861}, Aug. 2013.

\end{thebibliography}
\end{document}